\documentclass[journal,12pt,draftclsnofoot,onecolumn]{IEEEtran}

\usepackage{amssymb}
\usepackage{subfigure}
\usepackage{amsmath}
\usepackage{amsfonts}
\usepackage{mathrsfs}
\usepackage{engord}
\usepackage[dvips]{graphicx}
\usepackage{threeparttable}
\usepackage{booktabs}
\usepackage{epsfig}
\usepackage{color}
\usepackage{graphicx}
\usepackage{multirow}
\usepackage{extarrows}
\usepackage{epstopdf}
\usepackage{cite}
\begin{document}

\newtheorem{thm}{Theorem}
\newtheorem{prop}{Proposition}
\newtheorem{lem}{Lemma}
\newtheorem{method}{Coding method}
\newtheorem{remark}{Remark}[section]

\newcommand{\activenum}{K}
\newcommand{\groupIDset}{\mathcal{K}}
\newcommand{\usernum}{N}
\newcommand{\userIDset}{\mathcal{N}}
\newcommand{\timeset}{\mathcal{T}}
\newcommand{\activevector}{\mathbf{s}}
\newcommand{\activeset}{\mathcal{G}_{\activevector}}
\newcommand{\codebook}{\mathbf{X}}
\newcommand{\channeloutput}{\mathbf{y}}
\newcommand{\usercodeword}{\mathbf{x}}
\newcommand{\partition}{\mathbf{z}}
\newcommand{\sufficientbound}{C_1}
\newcommand{\necessarybound}{C_2}
\newcommand{\partitionset}{\mathbb{Z}}
\newcommand{\activesetcollection}{\mathbb{S}}
\newcommand{\hyper}{\mathcal{H}}
\newcommand{\hyperone}{\mathcal{H}_T^{\rq{}}}
\newcommand{\hypertwo}{\mathcal{H}^{*}_T}
\newcommand{\realactiveset}{\mathcal{G}_{\activevector_0}}

%
\title{Asymptotic Error Free Partitioning over Noisy Boolean Multiaccess Channels }


%
\author{\IEEEauthorblockN{Shuhang Wu,
Shuangqing Wei,
Yue Wang, 
Ramachandran Vaidyanathan and
Jian Yuan}}
\maketitle
\footnotetext[1]{S. Wu, Y. Wang and J. Yuan are with Department of Electronic Engineering, Tsinghua University, Beijing, P. R. China, 100084. (E-mail: {wsh05}@mails.tsinghua.edu.cn; {wangyue, jyuan}@mail.tsinghua.edu.cn). S. Wei and R. Vaidyanathan are with the School of Electrical Engineering and Computer Science, Louisiana State University, Baton Rouge, LA 70803, USA (Email: {swei, vaidy}@lsu.edu). 

This paper was submitted in June 2014 to IEEE Transactions on Information Theory, and is under review now.}


\begin{abstract}
	In this paper, we consider the problem of partitioning active users in a manner that facilitates multi-access without collision. The setting is of a noisy, synchronous, Boolean, multi-access channel where $K$ active users (out of a total of $N$ users) seek to access. A solution to the partition problem places each of the $N$ users in one of $K$ groups (or blocks) such that no two active nodes are in the same block. We consider a simple, but non-trivial and illustrative case of $K=2$ active users and study the number of steps $T$ used to solve the partition problem. By random coding and a suboptimal decoding scheme, we show that for any $T\geq (C_1 +\xi_1)\log N$, where $C_1$ and $\xi_1$ are positive constants (independent of $N$), and $\xi_1$ can be arbitrary small, the partition problem can be solved with error probability $P_e^{(N)} \to 0$, for large $N$. Under the same scheme, we also bound $T$ from the other direction, establishing that, for any $T \leq (C_2 - \xi_2) \log N$, the error probability $P_e^{(N)} \to 1$ for large $N$; again $C_2$ and $\xi_2$ are constants and $\xi_2$ can be arbitrarily small. These bounds on the number of steps are lower than the tight achievable lower-bound in terms of $T \geq (C_g +\xi)\log N $ for group testing (in which all active users are identified, rather than just partitioned). Thus, partitioning may prove to be a more efficient approach for multi-access than group testing.
\end{abstract}

\begin{IEEEkeywords}
partition information, conflict resolution, strong coloring, noisy Boolean channel 
\end{IEEEkeywords}

\section{Introduction}
	For successful payload transmission in networks, resources are needed to coordinate among users. A simple example is that of a MAC protocol, in which active users coordinate to avoid collision in channel access. 

{\bf The Partition Problem:}
	One simple way to achieve this coordination is through the {\it partition problem} defined below. For integer $N \geq 1$, let $\userIDset = \{1, \ldots, N\}$ and for integer $2 \leq K \leq N$, let $\activeset = \{i_1, \ldots, i_K\} \subseteq \userIDset$. A solution to the partition problem is a $K$-partition\footnote{A $K$-partition $\Pi = \{\mathcal{B}_1, \ldots, \mathcal{B}_K\}$ of $\userIDset$ is a set of $K$ non-empty subsets of $\userIDset$ that satisfies the following conditions: (a) for all $1\leq i<j\leq K$, $\mathcal{B}_i \cap \mathcal{B}_j = \emptyset$ and (b) $\bigcup_{i=1}^K \mathcal{B}_i = \userIDset$.} $\Pi = \{\mathcal{B}_1, \ldots, \mathcal{B}_K\}$ of $\userIDset$ such that for any $1\leq k \leq K$, we have $|\mathcal{B}_i \cap \activeset|=1$. That is, every group (or block) $\mathcal{B}_i$ of $\Pi$ contains exactly one element of $\activeset$.

	One could represent the $K$-partition $\Pi$ as a function $z:\userIDset \to \groupIDset$, where $\groupIDset = \{1, \ldots, K\}$, such that $z(i) = j$ iff $i \in \mathcal{B}_j$. In this paper, we will represent this function by a vector $\mathbf{z} = [z_i]$ where $z_i=j \in \groupIDset$ iff $i \in \mathcal{B}_j$. Thus, a valid partition $\mathbf{z}$ for an active set $\activeset$ in the partition problem satisfies $\forall i, j \in \activeset$, $i \neq j \Rightarrow z_i \neq z_j$.

	Now consider a set of $N$ users from $\userIDset$ sharing a Boolean multi-access channel. Assume $K$ of these users from set $\activeset \subseteq \userIDset$ are active, seeking to access the channel. If the partition problem is solved and each active user $i \in \activeset$ knows its group number $z_i$, then the $K$ active users can successively access the channel exclusively in $K$ data rounds: active user $i$ accesses the channel in data round $z_i$. This operation is fundamental to MAC protocols. 

     Observe that solving the partition problem does not require an active user to know the identities of other active users (as in the case in group testing). Thus, the partition problem, while sufficient for multi-access, holds the promise of a more efficient solution than group testing. In this paper, we demonstrate this potential of the partition problem.

{\bf The Channel:}
	We consider a slotted, noisy, Boolean, multi-access channel shared by $N$ users from $\userIDset$ of which a set $\activeset$ of $K$ users are active. For each round $t \geq 1$ and user $i \in \userIDset$, let $x_{i,t} \in \{0,1\}$ be a flag such that user $i$ writes to the channel in round $t$ iff $i \in \activeset$ and $x_{i,t} = 1$. The flag $x_{i,t}$ is also used to construct a transmission matrix as explained later. In the noiseless case, the channel provides (during round $t$) a feedback $y_{0,t} = \bigvee_{i \in \activeset} x_{i,t}$. The effect of observation noise is to alter the channel feedback $y_{0,t}$ to $y_t$ as follows. If $y_{0,t}=0$ then $y_t = 1$ with probability $q_{10}$ and $y_t=0$ with probability $1-q_{10}$. If $y_{0,t}=1$ then $y_t=0$ with probability $q_{01}$ and $y_t=1$ with probability $1-q_{01}$. Fig. \ref{Pfig1} illustrates this model. The channel output for $T$ rounds of transmissions is denoted by $T$-elements vector $\mathbf{y} = [y_1,\ldots, y_T]^\top$.
\begin{figure}
\centering
\includegraphics[width=1.8in]{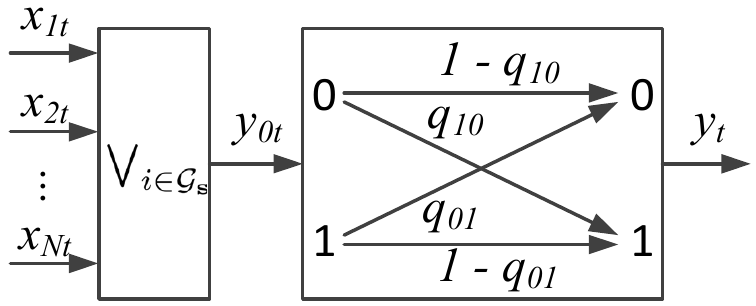}
\caption{Noisy multi-access Boolean channel.} \label{Pfig1}
\end{figure} 

	{\bf Our Approach and Main Results:}
    For any $T\geq 1$, $T$ rounds of potential transmissions by the $N$ users can be represented by an $N \times T$ transmission matrix $\mathbf{X} = [x_{i,t}]$; we say potential transmission to reiterate that it is only the active users $i$ with $x_{i,t}=1$ that transmit in round $t$. Recall that $\mathbf{z} = [z_i]$ where $1\leq i \leq N$ and $z_i \in \groupIDset$ denotes a partition of the $\userIDset$ and that $\mathbf{y} = [y_t]$ (where $1\leq t\leq T$) denotes the channel feedback over $T$ rounds. Let $\mathbb{Z}_{K;N} \triangleq \left\{\partition \in (\groupIDset)^N:\forall 1\leq k \leq K, \exists z_i = k\right\}$ denote the set of all possible $K$-partitions, consider a function $g: \{0,1\}^N \times \{0,1\}^T \times \{0,1\}^T \longrightarrow \mathbb{Z}_{K;N}$. For any given $N \times T$ transmission matrix $\mathbf{X}$ and a sequence of $T$ channel outputs $\mathbf{y}$, the function $g(\mathbf{X},\mathbf{y})$ produces a partition $\mathbf{z}$ of $\userIDset$.

	Given $T \geq 1$, the objective is to design an $N \times T$ transmission matrix $\mathbf{X}$ and a non-adaptive decoding function $g$ {\it a priori} such that for nearly every given $\activeset \subseteq \userIDset$, when written over $T$ rounds, the channel produces outputs $\mathbf{y}$ such that $g(\mathbf{X}, \mathbf{y}) = \mathbf{z}$ is a valid partition of the active users. We are seeking the infimum of $T$ over all possible $\mathbf{X}$, $g$ such that for any given $\activeset \subseteq \userIDset$, a valid partition $\mathbf{z}$ is produced, when $N \to \infty$. 

		In \cite{wsh2}, \cite{wu2014partition}, we have studied the noiseless channel case, using an i.i.d. Bernoulli random coding method to generate entries of $\mathbf{X}$. We modeled information gained towards partition construction as a sequence $\hyper_0, \hyper_1, \ldots, \hyper_T$ of hypergraphs. Here $\hyper_0$ is an $N$-nodes hypergraph with all possible hyperedges of rank $K$. Each channel output $y_t$ induces the removal of hyperedge(s) from $\hyper_{t-1}$ to construct $\hyper_t$. When the hyperedge corresponding to the active set $\activeset$ is in the resulting hypergraph $\hyper_T$, and every subgraph of $\hyper_T$ containing this hyperedge is strongly $K$-colorable (distinct colors within each hyperedge), then the active users have been partitioned (these colors being their group numbers). A point to note is that hypergraph $\hyper_t$ can be obtained from $\hyper_{t-1}$ using only the round-$t$ channel feedback $y_t$.

	For the noisy channel considered in this paper, we will still model partition problem from a perspective of a strong coloring of hypergraphs. However, $\hyper_t$, the hypergraph for round $t$, is constructed using the entire history of channel feedback $y_1, \ldots, y_t$. The analysis techniques used are also completely different from our earlier work\cite{wsh2}, \cite{wu2014partition}. We propose a sub-optimal strong typical set decoding method, and adopt random coding as well as a large deviation technique for an induced Markov chain. A more generalized structure is revealed than the extended Fibonacci numbers found in noiseless case, which could be potentially extended to solve more general cases with $K > 2$ active users. It is shown there is a gap between $\necessarybound$ and $\sufficientbound$ (the constants associated with the upper\&lower bounds on time $T$) under this scheme, which implies that there is room for further improvement.

	{\bf Prior Work:}
	The partitioning problem has a close relationship to conflict resolution\cite{1057022} and group testing\cite{ding2000combinatorial} (or compressed sensing\cite{malyutov2013search}) problems. Conflict resolution involves directly scheduling a transmission matrix $\codebook_{cr}$ (subscript ${cr}$ is for conflict resolution) for at least one slot so that each active user has its exclusive access to the channel. Note that the resulting transmission order of active nodes is not known to the users, only success of the transmission is ensured. Group testing also schedules a transmission matrix $\codebook_g$ such that set $\activeset$ of active users is exactly determined from the feedback $\channeloutput$; i.e., states of all users are identified. Our partition reservation system and group testing can both be used as a reservation step that assigns distinct transmitting orders to active users. Subsequent to this reservation stage, just $K$ slots of packets size are needed for active users to transmit their packets without conflict; in the reservation stage, however, the size of a time slot can be much smaller (just a bit). In contrast, the slots are of packet size during the entire process in the conflict resolution approach. We also establish that the partition reservation system needs less time than group testing, as it solves a weaker problem.

	To the best of our knowledge, Hajak first realized the nature of conflict resolution is to partition active users to different groups\cite{1056551,1056332}, and derived an achievable bound as partitioning information, without considering channel and transmission effects. The converse problem, which is close to a zero-error list-codes or perfect hashing problem is still an open problem; it was discussed by Hajak, K{\"o}rner, Arikan, {\it et al}, \cite{hajek1987conjectured,korner1988separating,2639, arikan1994upper}, and K{\"o}rner and Orlitsky in \cite[Chap. V]{720537}. These previous works on partition information are from the source coding perspective; i.e., representation of users' states using partition information. In contrast, we focus on construction of a partition relationship among active users by their explicit transmission over a collision and noisy Boolean multi-access channel. This problem has not been addressed previously. There are various approaches on non-adaptive conflict resolution and group testing for a Boolean multi-access channel; these methods are either \emph{combinatoric} or \emph{probabilistic}. These include overviews \cite{ding2000combinatorial}, \cite{gyorfilectures}, \cite{sandor2008}, and specific approaches including superimposed codes  \cite{kautz1964nonrandom, dyachkov1983survey, DeBonis2003223, chen2007exploring}, selective families \cite{Kowalski:2005:SPR:1073814.1073843}, broadcasting problem \cite{Clementi:2001:SFS:365411.365756}, and other methods \cite{capetanakis1979generalized, 1057020, sebHo1985two}. It should be noted that recently \cite{6157065} the group testing problem has been reformulated under an information theoretical framework to study the limits of restoration of IDs of active nodes over noisy Boolean multiple access channel. Noisy group testing is also discussed in Chan, {\it et al}. \cite{6120391}  and Malyutov \cite{5555301}.

	 Our work also has a significant impact on the understanding the limits of partitionability of interacting users in distributed systems. It has varies of applications. First, as stated before, it can be applied in the reservation stage of conflict resolution. Second, since the partition is obtained by all users, more complicated coordination is available for users to achieve better efficiency and more functions of the system. An example is that beside conflict resolution in time domain, the active users can avoid conflict in time-frequency domain, if they are assigned different orthogonal time-frequency codes according to the partition. Moreover, it could find use in other applications, including distributed multi-channel assignments, clustering, leader election, broadcasting, and resource allocation. \cite{Kowalski:2005:SPR:1073814.1073843,xavier1998introduction,vaidyanathan2003dynamic,hromkovic2005dissemination}.

	The rest of this paper is organized as follows. First, the problem is formulated in Section \ref{sec2}. A hypergraph strong coloring approach to decoding is presented in Section \ref{sec3}.  For the case with $K=2$ active users, the sufficient condition of time $T$ needed to obtain the desired partition is derived in Section \ref{sec4}, while the necessary condition under the same random coding and sub-optimal decoding framework is derived in Section \ref{sec5}. We compare our results with that of group testing in Section \ref{sec6}. Section \ref{sec7} concludes the results.

\section{System model and random coding}\label{sec2}
\subsection{Formulation}
	We further introduce some notation. In this paper, lower-case (resp., upper-case) boldface letters are used for column vectors (resp., matrices). For example, $\mathbf{w}=[w_i]$ denotes a vector with $w_i$ as the $i$-th element, while $\mathbf{W}=[w_{i,j}]$ denotes a matrix with element $w_{i,j}$ in row $i$ and column $j$. We use natural logarithms to base $e$. Symbols $\wedge$ and $\vee$ are used to represent AND, OR between events, for example, $B\wedge C$ denotes an event in which both $B$ and $C$ occur. These symbols are also used to represent logical AND and OR operations between Boolean operations, for example, $1\wedge 0 = 0$, $1 \vee 0 = 1$. The probability of a random variable $A$ having value $\tilde{A}$ is denoted by $p_A(\tilde{A}) \triangleq \text{Pr}(A=\tilde{A})$.  Similarly, $p_{A|B}(\tilde{A}|\tilde{B})\triangleq \text{Pr}(A=\tilde{A}|B=\tilde{B})$. Where there is no danger of ambiguity, we will drop the subscripts and simply write $p(A)$ or $p(A|B)$ to denote the above quantities.

	We assume that $K$ is known. We use a Boolean vector $\activevector=[s_1, \ldots, s_N]^{\top}$ to represent the active states of users, i.e., $s_i = 1$ iff $i \in \activeset$, recall that $\activeset = \{i_1, \ldots, i_K\}$ is the set of active users. Denote by $\activesetcollection_{K;N} \triangleq \{\activevector \in \{0,1\}^N:\sum s_i = K\}$ the set of all possible vectors $\activevector$ for $K$ active users. Active users use $T$ time slots to transmit according to $N \times T $ transmission matrix $\codebook$ and observe the feedback $\channeloutput$. Using these, the nodes obtain the $K$-partition $\partition=g(\codebook, \channeloutput)$. Assume that there is stationary, memoryless, observation noise in the channel under which the relation between $y_t$ (noisy channel feedback) and $y_{0t} =\bigvee_{i\in \activeset} x_{i,t}$ (noise-free channel feedback) is captured by the conditional probability $p_{y_t|y_{0t}}(y_t|y_{0t})$, as shown in Fig. \ref{Pfig1}. Recall that only active users $i$ with $x_{i,t}=1$ writes to the channel in round $t$. Therefore, the sequence of values collectively written to the channel (without the effect of noise) is $\channeloutput_0 = \codebook^\top\otimes \activevector \triangleq [\bigvee_{i} \left(x_{i,t} \wedge s_i\right)]$.  There are two dimensions in this problem, the user dimension $N$ and time dimension $T$. 

The partition problem can be illustrated by an example in Fig.~\ref{Pfig12}: user 1 and 2 are active, after transmission according to $\codebook$, the feedback $\channeloutput$ is observed instead of $\channeloutput_0$ due to the presence of noise; a common partition $[1~2~1~2]^{\top}$ is obtained by some decoding function $g$; it is a correct partition since active users are assigned to different groups. 

\begin{figure}
\centering
\includegraphics[width=3.2in]{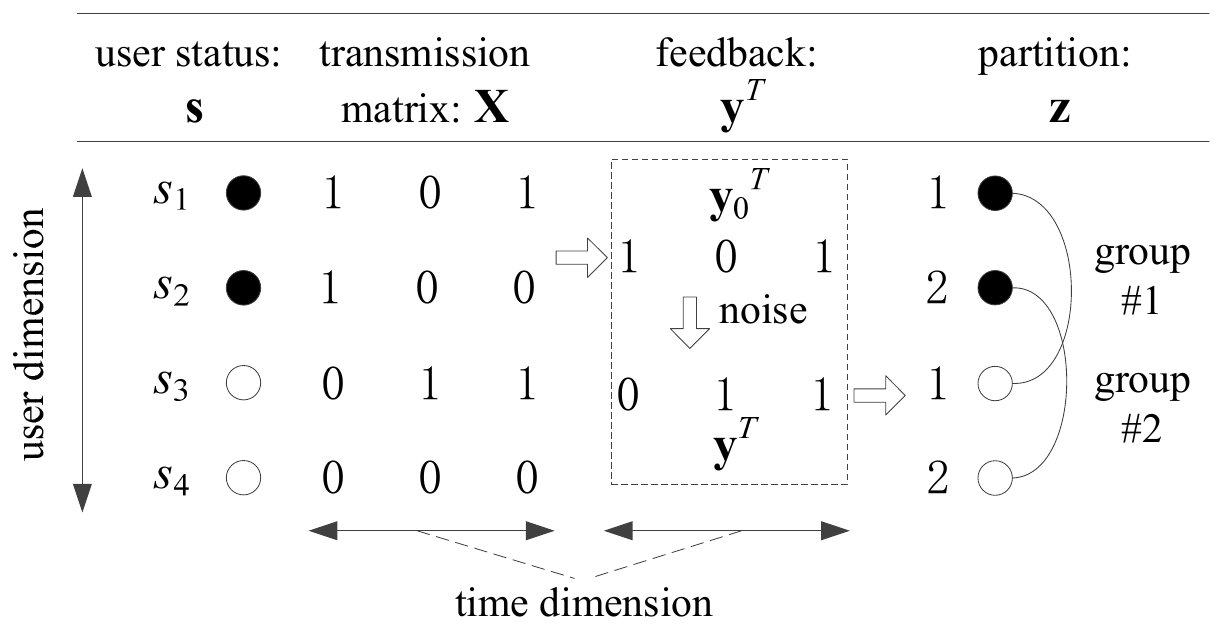}
\caption{Example of the formulation. ($N=4$, $K=2$, $\activeset=\{1,2\}$ means the \engordnumber{1} and \engordnumber{2} users are active, the number of time slots is $T=3$.)} \label{Pfig12}
\end{figure} 

	The partition problem can be treated as a coding problem in a multi-access channel from the information theoretic view as shown in Fig. \ref{pfig3}. Consider $N$ users with active states $\activevector$ as $N$ inputs to the system. The $i$-th row of $\codebook$, denoted by $\mathbf{x}_i^{\top}$($\mathbf{x}_i$ is a $T$ column vector), can be viewed as a codeword of user $i$, so that it will send $s_i \mathbf{x}_i^{\top}$ on the channel, and the feedback $\channeloutput$ is the output of channel.  A {\it distortion function} is defined for any
status vector $\activevector\in\activesetcollection_{K;N}$ and a partition vector $\partition \in \partitionset_{K;N}$ as follows:
\begin{align}
	d(\activevector,\partition) = 
	\begin{cases}
		0, &
			\mbox{if}~ \forall 1\leq i<j\leq N,
			~~~(s_i = s_j = 1)\Longrightarrow (z_i\neq z_j)\\
		1, &\rm{otherwise}
	\end{cases}.
	\label{distortion1}
\end{align}
The objective is to design a transmission matrix $\codebook$ (that produces channel output $\mathbf{y}$) and a corresponding decoding function $\partition = g(\mathbf{X}, \channeloutput)$, so that $d(\activevector,g(\mathbf{X}, \channeloutput))=0$. We use a probabilistic model to study the problem in this paper. Assume that every $\tilde{\activevector} \in \activesetcollection_{N;K}$ has the same probability $p_{\activevector}(\tilde{\activevector}) = 1/{N \choose K}$, consider the average error for a given $\codebook$ and $g$, defined by:
\begin{align}
	P^{(N)}_e(\codebook) \triangleq \sum_{\activevector \in \activesetcollection_{N;K}}p(\activevector) \sum_{\tilde{\channeloutput}} p_{\channeloutput|\channeloutput_0}(\tilde{\channeloutput}|\codebook^\top \otimes \activevector)  \mathbf{1}(d(\activevector, g(\mathbf{X},\tilde{\channeloutput}))\neq 0) \nonumber
\end{align}
where $\mathbf{1}(A)$ is the indicator function, whose value is $1$ when $A$ occurs and 0 otherwise; note that we use $p(\activevector)$ instead of $p_{\activevector}(\tilde{\activevector})$. Recall that we denote by $T$, the number of rounds over which the users transmit on the channel. We say ratio $\sufficientbound$ is achievable if when $\frac{T}{\log(N)} \geq \sufficientbound+\xi$ (for any constant $\xi>0$), there exists a matrix $\codebook^*$ and $g^*$ such that  $P^{(N)}_e(\codebook^*) \xlongrightarrow{N \to \infty} 0$. This is actually a sufficient condition in terms of the lowerbound of $\frac{T}{\log(N)}$ to attain $P^{(N)}_e(\codebook^*) \xlongrightarrow{N \to \infty} 0$. In Section \ref{sec4}, we will derive this achievable ratio for the partition problem.
\begin{figure}
\centering
\includegraphics[width=2.5in]{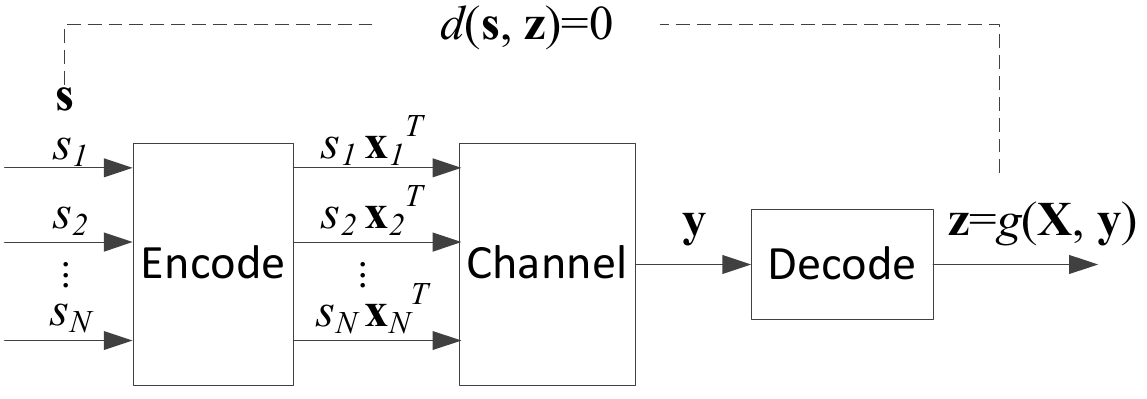}
\caption{Encoding-channel-decoding system with distortion criterion} \label{pfig3}
\end{figure}

\subsection{Random Coding}
	To find an achievable ratio $\sufficientbound$, we employ a random coding approach by generating each $x_{i,t}$ by independent and identical Bernoulli distribution with probability $p = \text{Pr}(x_{i,t}=1)$, and use the optimal Bayesian decoding method with risk function $P^{(N)}_e = \text{E}_{\mathbf{s}}[\mathbf{1}(d(\mathbf{s},g(\mathbf{X},\mathbf{y}))=0)]$; that is, for a realization of $\codebook$, when observing $\tilde{\channeloutput}$, we choose $\tilde{\partition}^* =g(\mathbf{X},\mathbf{\tilde{y}})$, so that $\tilde{\partition}^* = \arg\max_{\tilde{\partition} \in \partitionset_{K;N}} \text{E}_{\mathbf{s}}[\mathbf{1}(d(\mathbf{s},\tilde{\partition})=0)|
\tilde{\mathbf{y}}]$, with proper normalization one may have:
\begin{align}
	 \tilde{\partition}^* = \arg\max_{\tilde{\partition} \in \partitionset_{K;N}} W_{\tilde{\channeloutput};\codebook}(\tilde{\partition})
\end{align}
where 
\begin{align}
	W_{\tilde{\channeloutput};\codebook}(\tilde{\partition}) \triangleq \sum_{\activevector \in \activesetcollection_{N;K}} \mathbf{1}\left(d(\activevector,\tilde{\partition}) = 0\right)p_{\channeloutput|\channeloutput_0}(\tilde{\channeloutput}|\codebook^\top\otimes \activevector) \label{MAPdecode}
\end{align}
If there is more than one $\tilde{\partition}^*$ with the maximum value, we select any one. For notation simplicity, we will use $p(\channeloutput|\codebook^\top\otimes \activevector)$ instead of $p_{\channeloutput|\channeloutput_0}(\tilde{\channeloutput}|\codebook^\top\otimes \activevector)$ in the remainder of this paper. Similarly, we use $p(\channeloutput,\codebook^\top\otimes \activevector)$ instead of $p_{\channeloutput,\channeloutput_0}(\tilde{\channeloutput},\codebook^\top\otimes \activevector)$.

	Then, the average error over all realizations of $\codebook$ is
\begin{align}
	P^{(N)}_e \triangleq &\sum_{\codebook}Q(\codebook)P^{(N)}_e(\codebook)\nonumber\\
\overset{(a)}{=}&\sum_{\codebook}Q(\codebook)\sum_{\channeloutput}p(\channeloutput|\codebook^\top \otimes \activevector_0)\mathbf{1}(d(\activevector_0, g(\mathbf{X},\channeloutput))\neq 0), \label{averageerror}
\end{align}
where $Q(\codebook)$ denotes the probability of generating $\codebook$. Equality $(a)$ of Eq. \eqref{averageerror} is due to the symmetry in the generation of $\codebook$, where any particular $\activevector_0$ can be chosen as an input for our analysis. We will assume corresponding active set $\realactiveset=\{1,\ldots,K\}$ in the rest of the paper. 	Denote by $P_e^{(\infty)}$ the asymptotic value of $P_e^{(N)}$. Since if $P_e^{(\infty)} = 0$, there must exist an $\codebook^*$ with  $P^{(N)}_e(\codebook^*) \xlongrightarrow {N \to \infty} 0$. We will seek a $\sufficientbound$ so that when $T/(\log N) \geq \sufficientbound + \xi$ for any $\xi>0$, then $P_e^{(\infty)}= 0$. 
The optimal Bayesian decoding is quite complicated to analyze. In next section, we propose a sub-optimal decoding method to analyze the average error probability from a strong hypergraph coloring perspective.

\section{A Graph Coloring Approach for Decoding}\label{sec3}
	A hypergraph decoding approach is presented in this section. To better understand this method, all examples are given in $K=2$ case, in which a hypergraph becomes a graph. 
\subsection{Hypergraph view}
	Our decoding method is illustrated in Fig. \ref{pfig5}. For a given input $\activevector_0$, the channel output $\channeloutput$ is observed. This output could be different from $\channeloutput_0 = \codebook^\top \otimes \activevector_0$ due to the presence of noise. A possible set of active users  with a \lq\lq{}sufficiently large probability\rq\rq{} generating $\channeloutput$ can be inferred. Let
\begin{align}
	\activesetcollection_{\channeloutput} = \left\{\activevector\in \activesetcollection_{N;K}: p(\channeloutput|\codebook^\top \otimes \activevector)~\text{is sufficiently large}\right\}
\end{align} 
be the set of all such sets of active users. The idea is to select a $\partition = g(\mathbf{X},\mathbf{y})$ such that $\activesetcollection_{\channeloutput} \cap \{\mathbf{s}:d(\activevector,\partition)=0\}$ is maximized.

	For any given active set represented as vector $\mathbf{s}_0$,	we now outlines the stages used model the transmission and observation of channel feedback, and the construction of a partition $\mathbf{z}$ corresponding to $\mathbf{s}_0$.

\begin{figure}
\centering
\includegraphics[width=3.2in]{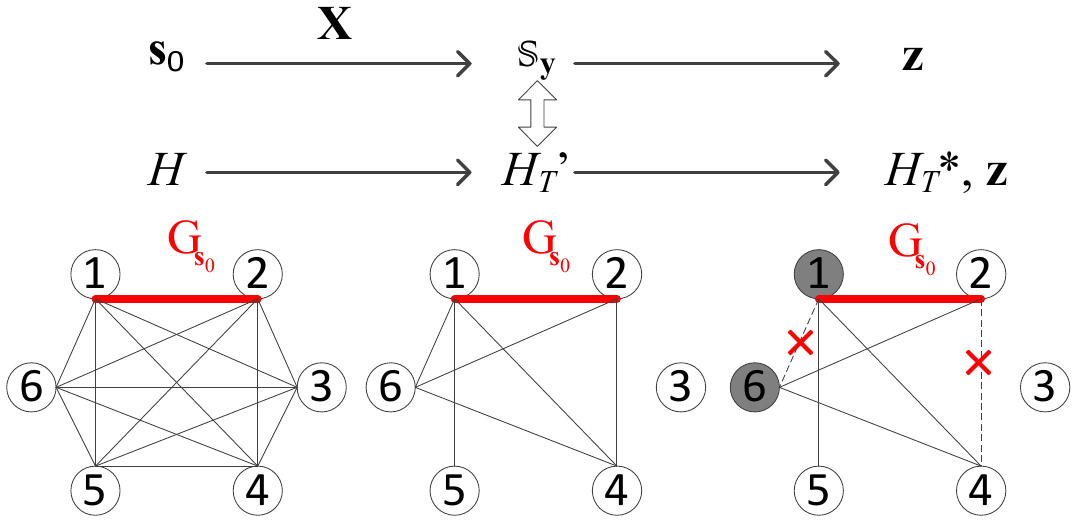}
\caption{An example of the graph view ($N=6$, $K=2$).} \label{pfig5}
\end{figure}

\begin{enumerate}
	\item \textbf{Source}: The input $\mathbf{s}$ is represented initially as an empty hypergraph $\mathcal{H}_0 = (\userIDset, \emptyset)$ with $N$ nodes and no hyperedges. This reflects our initial lack of knowledge about the active users.

	\item \textbf{Transmission and observation}: As the active users transmit over the channel using a transmission matrix $\mathbf{X}$ and the channel provides feedback $\mathbf{y}$, we modify $\mathcal{H}_0$ by adding $K$-elements hyperedges according to active sets with \lq\lq{}sufficiently large probability\rq\rq{}. Let the new hypergraph be $\hyperone = (\userIDset, E^{\rq{}}_T)$.

	\item \textbf{Partition}: After $T$ rounds the information obtained so for is in a hypergraph $\hyperone = (\userIDset,E^{\rq{}}_T)$, as explained before. The decoding process seeks to find the largest sub-hypergraph $\hypertwo \subseteq \hyperone$ (where $\hypertwo = (\userIDset,E^{*}_T)$ and $E^{*}_T \subseteq E^{\rq{}}_T$) such that $\hypertwo$ is strongly $K$-colorable\footnote{A hypergraph is strongly $K$-colorable iff there exists an arrangement of $K$ colors to nodes, such that no hyperedges contains two nodes of the same color.}. Observe that the $K$-coloring of $\hypertwo$ gives a $K$-partition, $\mathbf{z}$ of active vector $\mathbf{s}_0$.
\end{enumerate}

	It must be pointed out that the above hypergraph $\hypertwo$ may not correctly partition $\userIDset$ with respect to $\realactiveset$. This error is due to either (a) the hyperedge $\realactiveset$ not being present in $\hyperone$ or (b) $\realactiveset$ being deleted from $\hyperone$ to generate $\hypertwo$. However if $T/\log N$ is sufficiently large (see Section \ref{sec4}), we prove that the decoding is asymptotically error-free.

	Thus, the process can be represented as $\mathcal{H}_0 \to \hyperone \to \hypertwo, \partition$ corresponding to the expression from vectors $\activevector_0 \to \activesetcollection_{\channeloutput} \to \partition$, as shown in Fig. \ref{pfig5} by an example of $N=6$, $K=2$. Compared with group testing whose objective is to obtain $\hypertwo=\hyperone$ with only one hyperedge $\realactiveset$, our partition problem allows $\hyperone$ and $\hypertwo$ to have more hyperedges to be added, so less effort is needed.  This translates to higher achievable rate than that of the group testing problem. The objective is then to design an efficient $\codebook$, which essentially construct such a $\hyperone$ from which we can correctly obtain $\hypertwo$ and $\partition$. 

	An observation of the decoding method is that if the real edge $\realactiveset \in E^*_T$, definitely we will get a $\partition$ so that $d(\activevector_0,\partition)=0$; otherwise we may get a wrong partition. Since $\hypertwo \subseteq \hyperone$, we also need $\realactiveset \in E^{\rq{}}_T$. 

	As noted earlier, unlike the noiseless case \cite{wsh2, wu2014partition}, the hypergraph $\hyperone$ cannot be constructed sequentially for the noisy channel. (In noiseless case, $\mathcal{H}_0$ can be seen as a complete hypergraph with all $K$-element hyperedges. For any hyperedge $\activeset$, it will be deleted from $\mathcal{H}_{t-1}$ to obtain $\mathcal{H}_{t}$ if either condition is satisfied at round $t$: when $y_t=1$, no users in $\activeset$ transmits 1; or  when $y_t=0$, at least one of the users in $\activeset$ transmit 0.) The key reason is that in noiseless case, we have $\activeset \in E^{\rq{}}_T \iff \activevector \in \activesetcollection_{\channeloutput}$, (note that $\mathbf{y} = \mathbf{y}_0$, and $\mathbf{y}_0 = \mathbf{X}^\top \otimes \mathbf{s}_0$), so we always have the real hyperedge $\realactiveset$ in $\hyperone$; but in noisy case this is not satisfied, so that $\realactiveset$ may be not in $\hyperone$. Thus, the aim is to ensure for given $\mathbf{y}$, the generated hypergraph $\hyperone = (\userIDset,E^{\rq{}}_T)$ satisfies for any input $\activevector$,
\begin{align}
\activeset \in E^{\rq{}}_T \iff \activevector \in \activesetcollection_{\channeloutput_0}. \label{idealcondition}
\end{align}
asymptotically. 

\subsection{A Suboptimal decoding method}
	By adopting the strong typical set decoding approach \cite{csiszar2011information}, we can develop a joint edge construction method. Define a strong typical set $\mathcal{E}_{\epsilon}^T$, for any $a\triangleq (w,w_0) \in \{0,1\}^2$, and a small constant $\epsilon >0$,
\begin{align}
	\mathcal{E}_{\epsilon}^T = 
\left\{\begin{array}{ll}
\left[\tilde{\channeloutput}, \tilde{\channeloutput}_0\right]\in \left(\{0,1\}^2\right)^{T}:&\\
\left|\frac{1}{T}N\left(a|\left[\tilde{\channeloutput}, \tilde{\channeloutput}_0\right]\right)-p_{y,y_0}(a)\right|\leq\frac{\epsilon}{4}, &\text{if}~p_{y,y_0}(a)>0\\
N\left(a|\left[\tilde{\channeloutput}, \tilde{\channeloutput}_0\right]\right)=0,&\text{if}~p_{y,y_0}(a)=0
\end{array}\right\}\label{rules}
\end{align} 
where for any collection of $L$ Boolean $T$-bit vectors $[\mathbf{w}_1, \ldots, \mathbf{w}_L] \in \left(\{0,1\}^L\right)^T$, and a pattern $a \in \{0,1\}^L$, $N\left(a|[\mathbf{w}_1, \ldots, \mathbf{w}_L]\right)$ denotes the number of times of having pattern $a$ in the sequence $\{(w_{1,t}, \ldots, w_{L,t})\}_{t=1}^T$. And more specially,
\begin{align}
p_{y,y_0}(a) = p_{y_t,y_{0t}}(w,w_0) = 
\begin{cases}
p_{y_t|y_{0t}}(w|1)(1-(1-p)^K), &w_0 = 1\\
p_{y_t|y_{0t}}(w|0)(1-p)^K, &w_0 = 0
\end{cases}\label{yy0}
\end{align}
Thus, $\mathcal{E}_{\epsilon}^T$ is the strong typical set that sample frequencies are close to the true probability values.

The joint edge construction is first, choose a small $\epsilon >0$; then for a given $\codebook$, construct $\hyperone$ by:
\begin{align}
	\activeset \in E^{\rq{}}_T ~~\text{iff}~~ \left[\channeloutput, \bigvee_{i \in \activeset}\mathbf{x}_i\right] \in \mathcal{E}_{\epsilon}^T \label{contitionnoise}
\end{align}
Because of the feature of strong typical set, when $T\to \infty$, intuitively \eqref{idealcondition} almost surely holds for all $\codebook$ and the resulting $\channeloutput$, thus we will use this joint criteria to construct $\hyperone$ instead of the sequential method in the noiseless case. It\rq{}s a general method which is equivalent to the sequential approach in the absence of noise.

	The following steps express the action performed collectively by the active users in partitioning $\activeset$ with $\mathbf{y}$. We note that this is not an algorithm, just an illustration of the functional steps of transmission and decoding.

\emph{Joint edges construction decoding:}
\begin{enumerate}
\item When observing $\channeloutput$, all users construct $\hyperone$ by the rule that $\activeset \in E^{\rq{}}_T ~~\text{iff}~~ \left[\channeloutput, \bigvee_{i \in \activeset}\mathbf{x}_i\right] \in \mathcal{E}_{\epsilon}^T$;
\item Determine $\hypertwo \subseteq \hyperone$ by deleting the minimum number of hyperedges from $\hyperone$ such that $\hypertwo$ is $K$-strongly colorable, the output $\partition$ is a $K$-strong colouring of $\hypertwo$. 
\end{enumerate}

\begin{remark} 
	Optimal Bayesian decoding tries to find the maximum $\partition$ for $W_{\channeloutput;\codebook}(\partition)$ defined in \eqref{MAPdecode}, while our proposed joint method try to maximize:
\begin{align}
	W_{\channeloutput;\codebook}\rq{}(\partition) \triangleq \sum_{\activevector \in \activesetcollection_{N;K}} \mathbf{1}\left(d(\activevector,\partition) = 0\right)\mathbf{1}\left(\left[\channeloutput, \bigvee_{i \in \activeset}\mathbf{x}_i\right] \in \mathcal{E}_{\epsilon}^T\right)
\end{align}
i.e., we replace the weight $p(\channeloutput|\codebook^\top\otimes\activevector)$ in Bayesian decoding by quantizing it into $\{0,1\}$ according to the strong typical set $\mathcal{E}_{\epsilon}^T$. Thus, the proposed method is optimal when there is no noise, but suboptimal in the presence of noise. However, it has explicit geometric meaning in a hypergraph view which further enables us to derive an achievable bound based on an induced Markov chain, as shown in next two sections.  
\end{remark}

\subsection{Simplification for $K=2$ case}
	Before giving the main result of $K=2$ case, we first observe that the hypergraph $\hyper_0$, $\hyperone$, $\hypertwo$ are now graphs. We will provide a sub-optimal analysis to further simplify the calculation of $P_e^{(N)}$. The proposed decoding method includes two steps: to construct $\hyperone$; to find $\hypertwo$ and $\partition$. First, $a \in \{0,1\}^2$  imports four constraints on the construction of $\hyperone$ using $\mathcal{E}_{\epsilon}^T$ (see Eq. \eqref{rules}). These constraints are correlated, since $\sum_{w,w_0}N\left((w,w_0)|\left[\channeloutput, \bigvee_{i \in \activeset}\mathbf{x}_i\right]\right) = T$. So we will reduce the number of constraints by selecting only two of them $a = (w,w_0)$ with $w \in \{0,1\}, w_0 = 0$ as the constraints. Strictly speaking, assuming $0<p_{y,y_0}(a)<1$, since $p_y(w) =p_{y,y_0}(w,1)+p_{y,y_0}(w,0)$,
we have the following sufficient constraints for Eq. \eqref{rules} by selecting $\tilde{\epsilon} = 2 \epsilon$,
\begin{align}
	&(i,j) \in E^{\rq{}}_T \iff [\channeloutput, \mathbf{x}_i \vee \mathbf{x}_j] \in \mathcal{E}_{\epsilon}^T\nonumber\\
\Longrightarrow &\left|\frac{1}{T}N(w|\channeloutput) - p_{y}(w)\right| \leq \tilde{\epsilon}/4~\text{and} \nonumber\\
&\left|\frac{1}{T}N((1,0)|[\channeloutput, \mathbf{x}_i\vee\mathbf{x}_j])-p_{y,y_0}(1,0)\right|\leq \tilde{\epsilon}/4 ~\text{and} \nonumber\\
&\left|\frac{1}{T}N((0,0)|[\channeloutput, \mathbf{x}_i\vee \mathbf{x}_j])-p_{y,y_0}(0,0)\right|\leq \tilde{\epsilon}/4, \label{muij1}
\end{align}
and necessary constraints for Eq. \eqref{rules} with $\hat{\epsilon} = \frac{1}{2} \epsilon$,
\begin{align}
	&(i,j) \in E^{\rq{}}_T \iff [\channeloutput, \mathbf{x}_i \vee \mathbf{x}_j] \in \mathcal{E}_{\epsilon}^T\nonumber\\
\Longleftarrow &\left|\frac{1}{T}N(w|\channeloutput) - p_{y}(w)\right| \leq \hat{\epsilon}/4~\text{and} \nonumber\\
&\left|\frac{1}{T}N((w,0)|[\channeloutput, \mathbf{x}_i\vee\mathbf{x}_j])-p_{y,y_0}(w,0)\right|\leq \hat{\epsilon}/4, \forall w\in\{0,1\}, \label{muij2}
\end{align}
We will use these two constraints to analyze the sufficiency or necessity instead of Eq. \eqref{rules}. Actually these two constraints are approximately equivalent (up to a constant multiplying $\epsilon$). As we are primarily interested in the conditions under which $P_e^{(N)} \to 0$, rather than how fast it approaches 0. Once $\channeloutput$ is given, we can divide time slots into two blocks $\timeset^w \triangleq \{t:y_t=w\}, w \in \{0,1\}$ based on the value of $y_t$, and define $T^w = |\timeset^w|$ as the number of slots in each block; we also separate each codeword $\mathbf{x}_i$ into two blocks: $\mathbf{x}_i^w = [x_{i,t}]_{\{t\in\timeset^w\}}$, $w\in\{0,1\}$ according to the indices of $[y_t]_{t \in \timeset^w}$, and $N((w,0)|[\channeloutput, \mathbf{x}_i\vee \mathbf{x}_j]) = N(0|[ \mathbf{x}_i^w\vee \mathbf{x}_j^w])=N((0,0)|[ \mathbf{x}_i^w, \mathbf{x}_j^w])$, because $x_{i,t}^w \vee x_{j,t}^w = 0$ iff $x_{i,t}^w = 0,x_{j,t}^w = 0$. Thus we have:
\begin{align}
	(i,j) \in E^{\rq{}}_T \Longrightarrow \channeloutput \in \mathcal{E}^{T}_{y,\tilde{\epsilon}} ~\text{and}~ [\mathbf{x}^1_{i}, \mathbf{x}^1_j] \in \mathcal{E}^{T^1}_{1,\tilde{\epsilon}} ~\text{and}~[\mathbf{x}^0_{i},\mathbf{x}^0_{j}] \in \mathcal{E}^{T^0}_{0,\tilde{\epsilon}} \label{muij}
\end{align} 
where
\begin{align}
	\mathcal{E}^{T}_{y,\tilde{\epsilon}} = &\left\{\tilde{\channeloutput}\in \{0,1\}^{T}:
 \left|\frac{1}{T}N((0,0)|\tilde{\channeloutput})-p_{y}(w)\right|\leq \tilde{\epsilon}/4\right\}\nonumber\\
	\mathcal{E}^{T^w}_{w,\tilde{\epsilon}} = &\Bigg\{[\tilde{\mathbf{x}}_1,\tilde{\mathbf{x}}_2]\in \left(\{0,1\}^2\right)^{T^w}:
 \left|\frac{1}{T}N((0,0)|[\tilde{\mathbf{x}}_1,\tilde{\mathbf{x}}_2])-p_{y,y_0}(w,0)\right|\leq \tilde{\epsilon}/4\Bigg\},~w\in\{0,1\}\label{newrule1}
\end{align}
We note that the implication of Eq. \eqref{newrule1} can be reversed with $\tilde{\epsilon} = \hat{\epsilon}$ (as in Eq. \eqref{muij1} and \eqref{muij2}).
It means given $\channeloutput$, we can separately check $\channeloutput$ and the codewords in $\timeset^1$ and $\timeset^0$ by \eqref{muij} to construct $\hyperone$. Note that by Eq. \eqref{yy0} we have $p_{y,y_0}(1,0) = (1-p)^2q_{10}$, $p_{y,y_0}(0,0) = (1-p)^2(1-q_{10})$.

\begin{remark}
	 The physical meaning of \eqref{muij} is that given $\channeloutput$ satisfying \eqref{muij}, for each edge $(i,j)$, in block $\timeset^1$, we count the number of times that $(y_{t}, x_{i,t} \vee x_{j,t}) = (1,0)$, and in block $\timeset^0$ count the number times that $(y_{t}, x_{i,t} \vee x_{j,t}) = (0,1)$. If they are close to $p_{y,y_0}(1,0)T$ and $p_{y,y_0}(0,1)T$ (note that $p_{y,y_0}(0,1)T=(p_y(0)-p_{y,y_0}(0,0))T$), the edge $(i,j)$ is in $\hyperone$, otherwise not.  This is shown in Fig. \ref{pfig14}. Note that when there is no noise, these two numbers should be zero, that\rq{}s why a sequential construction can be used without noise, and why the problem is more difficult to solve in the presence of noise and we need to resort to a large deviation technique.
\end{remark}

\begin{figure}
\centering
\includegraphics[width=3.5in]{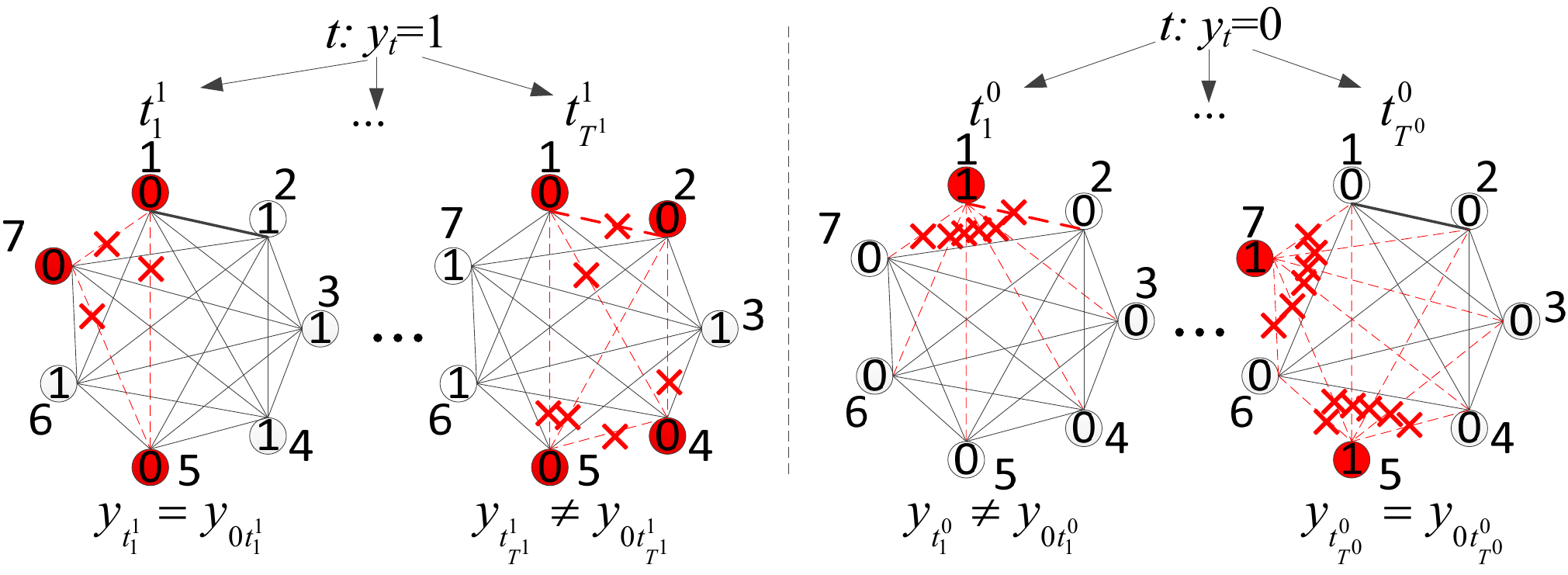}
\caption{An example of constructing edges in $\timeset^1$ and $\timeset^0$. Here $N=7$, and in the cycle the value of $x_{i,t}$ is shown. In some time, the codewords of the real edge $\{0,1\}$ might have $y_{t} \neq x_{1,t} \vee x_{2,t}$ because of the noise. We should count all such times in each block to decide if an edge is in $\hyperone$ or not in noisy case.}  \label{pfig14}
\end{figure}

	Another simplification of the $K=2$ case is in finding $\hypertwo$ and $\partition$. If $\partition$ is correct, we have the following sufficient conditions to achieve an acceptable partition:
\begin{align}
	d(\activevector_0,\partition)=0 \overset{(a)}{\Longleftarrow} ~&\realactiveset \in \hypertwo\nonumber\\
 \overset{(b)}{\Longleftarrow}  ~&\realactiveset \in \hyperone, ~\hyperone~\text{contains no 1-odd cycles}\nonumber
\end{align}
Eq. $(a)$ is obvious since $\hypertwo$ is 2-colorable. For Eq. $(b)$, if $K=2$, $\hypertwo$ is 2-colorable iff $\hypertwo$ contains no odd cycles. Now consider one type of odd cycles that contains $\realactiveset$ as an edge, called first type of odd cycles, denoted as \lq\lq{}1-odd cycles\rq\rq{}. Because all odd cycles in $\hyperone$ should be broken by deleting least edges to get $\hypertwo$, if there are no 1-odd cycles in $\hyperone$, $\realactiveset$ won\rq{}t be deleted and will be in $\hypertwo$. Thus, to find an achievable rounds of test, we deem decoding is correct iff $\realactiveset \in \hyperone$ and $\hyperone$ contains no 1-odd cycles for simplification, which is a sub-optimal analysis as belows.

\emph{Sub-optimal analysis for joint edges construction decoding:}
\begin{enumerate}
\item When observing $\channeloutput$, users construct $\hyperone$ by the rule that $[\channeloutput, \mathbf{x}_i \vee \mathbf{x}_j] \in \mathcal{E}_{\epsilon}^T$, which can be further simplified to $\channeloutput \in \mathcal{E}^{T}_{y,\tilde{\epsilon}} ~\text{and}~ [\mathbf{x}^1_{i}, \mathbf{x}^1_j] \in \mathcal{E}^{T^1}_{1,\tilde{\epsilon}} ~\text{and}~[\mathbf{x}^0_{i},\mathbf{x}^0_{j}] \in \mathcal{E}^{T^0}_{0,\tilde{\epsilon}}$ when considering sufficiency, or $\channeloutput \in \mathcal{E}^{T}_{y,\hat{\epsilon}} ~\text{and}~ [\mathbf{x}^1_{i}, \mathbf{x}^1_j] \in \mathcal{E}^{T^1}_{1,\hat{\epsilon}} ~\text{and}~[\mathbf{x}^0_{i},\mathbf{x}^0_{j}] \in \mathcal{E}^{T^0}_{0,\hat{\epsilon}}$ when considering necessity;
\item If $\realactiveset \in \hyperone$ and $\hyperone$ contains no 1-odd cycles, the decoding output is correct; otherwise it is wrong.
\end{enumerate}

	Since $\activevector_0$ is not known in advance, the second step can\rq{}t be used in application, thus this method is only used for analysis. In the next section we will derive a sufficient condition on $\frac{T}{\log N}$ to achieve $P_e^{(\infty)}= 0$ by this analysis; then in Section \ref{sec5}, we derive the necessary condition under which $\hyperone$ will definitely have 1-odd cycles in the framework of Bernoulli random coding and this sub-optimal analysis, which shows the limits of performance of this method.

\section{Main result: Sufficient condition for $K=2$ case} \label{sec4}
\begin{thm}\label{thm1}
	\textit{When $K=2$, if $\frac{T}{\log N} \geq \sufficientbound +\xi$ for any constant $\xi>0$, we have $P^{(\infty)}_e = 0$, where $\sufficientbound\triangleq 1/\max_{p} C(p,q_{10},q_{01})$, and} 
\begin{align}
	C(p,q_{10},q_{01}) = p_1\varphi_1 + p_0 \varphi_0, \label{rate}
\end{align}
\textit{$p_1 \triangleq p_{y}(1)=(1-(1-p)^2)(1-q_{01})+(1-p)^2q_{10}$, $p_0 \triangleq p_y(0)= 1-p_1$, and}
	\begin{align}
		\varphi_1 = \sup_{\lambda \in \mathbf{R}} \left(\lambda \frac{(1-p)^2q_{10}}{p_1} - \log \rho_+\right), ~
		\varphi_0 = \sup_{\lambda \in \mathbf{R}} \left(\lambda \frac{(1-p)^2(1-q_{10})}{p_0} - \log \rho_+\right)\nonumber\\
\rho_{+} = \frac{1}{2}(p+(1-p)e^{\lambda}+\sqrt{(p-(1-p)e^{\lambda})^2+4p(1-p)}).\nonumber
	\end{align}
\end{thm}

We can check when $q_{10}=q_{01}=0$, the result is exactly the same as that in the noiseless case in \cite{wsh2, wu2014partition}, and when $q_{10}+q_{01} =1$, or $p=0, 1$, $C(p,q_{10},q_{01})=0$, which means we can\rq{}t recover the partition with vanishing average error no matter how large $T$ is, since in this case $\channeloutput$ is independent of $\activevector$. Moreover, it can be seen if in another system $(q_{01}\rq{},q_{10}\rq{})=(1-q_{01},1-q_{10})$, the corresponding $p_w\rq{} = p_{1-w}$, $\varphi_w\rq{} = \varphi_{1-w}$, $\forall w\in\{0,1\}$, thus $C(p,q_{10},q_{01})$ and $\sufficientbound$ is symmetric with the center $(q_{01},q_{10})=(0.5,0.5)$. 

The complete proof is in Appendix \ref{appendix.a}. Below is a sketch of the main ideas in the proof.

	1) Assume the real edge $\realactiveset = \{1,2\}$, in order to calculate the error probability easily, we will consider $[\channeloutput, \mathbf{x}_1, \mathbf{x}_2]$ to be in a strong typical set through the proof:
\begin{align}
	\mathcal{A}_{\delta}^T = 
\left\{\begin{array}{ll}
\left[\tilde{\channeloutput}, \tilde{\mathbf{x}}_1,\tilde{\mathbf{x}}_2\right]\in \left(\{0,1\}^3\right)^{T}:&\\
\left|\frac{1}{T}N\left(a|\left[\tilde{\channeloutput}, \tilde{\mathbf{x}}_1,\tilde{\mathbf{x}}_2\right]\right)-p_{y,x_1,x_2}(a)\right|\leq\frac{\delta}{16}, &p_{y,x_1,x_2}(a)>0\\
N\left(a|\left[\tilde{\channeloutput}, \tilde{\mathbf{x}}_1,\tilde{\mathbf{x}}_2\right]\right)=0,&p_{y,x_1,x_2}(a)=0
\end{array}\right\}, \label{atypical}
\end{align} 
where $a\triangleq(w,u,v) \in \{0,1\}^3$, and $p_{y,x_1,x_2}(w,u,v) = p_{y,y_0}(w,u\vee v)p_x(u)p_x(v)$, $p_x(u) \triangleq \text{Pr}(x_{i,t} =u)$ is the Bernoulli pdf. Note that it is different from Eq. \eqref{rules} which is defined on $[\channeloutput, \mathbf{x}_i\vee \mathbf{x}_j]$ for any edge $(i,j)$, since we need more restriction on the codewords of the real edge to simplify analysis. If we choose $\delta = \tilde{\epsilon}$, by the sufficient condition of Eq. \eqref{muij}, we have $[\channeloutput, \mathbf{x}_1, \mathbf{x}_2] \in \mathcal{A}_{\tilde{\epsilon}}^T \Longrightarrow \{1,2\} \in E^{\rq{}}_T$ and $\channeloutput \in \mathcal{E}^{T}_{y,\tilde{\epsilon}}$, which makes the analysis easier. For simplification of the notation, we will use $\epsilon$ instead of $\tilde{\epsilon}$ in the rest of this section. Since for any event $A$, 
\begin{align}
	\text{Pr}(A) = &\text{Pr}(A, [\channeloutput, \mathbf{x}_1, \mathbf{x}_2] \in \mathcal{A}_{\epsilon}^T) + \text{Pr}(A, [\channeloutput, \mathbf{x}_1, \mathbf{x}_2] \notin \mathcal{A}_{\epsilon}^T)\nonumber\\
\leq &\max_{[\channeloutput, \mathbf{x}_1, \mathbf{x}_2] \in \mathcal{A}_{\epsilon}^T}\text{Pr}(A|[\channeloutput, \mathbf{x}_1, \mathbf{x}_2]) + \text{Pr}([\channeloutput, \mathbf{x}_1, \mathbf{x}_2] \notin \mathcal{A}_{\epsilon}^T)\nonumber
\end{align}
Based on the feature of strong typical set, $\text{Pr}([\channeloutput, \mathbf{x}_1, \mathbf{x}_2] \notin \mathcal{A}_{\epsilon}^T) \to 0$, as $T \to \infty$; and for the maximum, when $\epsilon$ is small, every element in the typical set is nearly the same. So asymptotically $[\channeloutput, \mathbf{x}_1, \mathbf{x}_2] \in \mathcal{A}_{\epsilon}^T$, we consider this condition is held in the following parts, and calculate the probability conditioning on a given $[\channeloutput, \mathbf{x}_1, \mathbf{x}_2]$.  Since in this case $\{1,2\}\in E^{\rq{}}_T$, we just need to consider the probability that $\hyperone$ contains 1-odd cycles. Assume the probability of existence of a particular 1-odd cycle of $M$ vertices in $\hyperone$ to be $p_{e;M}$, there are ${N-2 \choose M-2} (M-2)! \leq N^{M-2}$ such odd cycles and all of them are equiprobable due to the symmetry in generating $\codebook$. 
Thus, by union bound, we have
\begin{align}
	P_e^{(N)} \leq \sum_{M\geq 3, M~\text{is odd}} e^{(M-2)\log N} p_{e;M}\label{exp1}
\end{align}
By the physical meaning stated in Remark 3.2, we will consider if this particular odd cycle will be constructed in block $\timeset^1$ or $\timeset^0$, i.e., $p_{e;M} = \mu_{M}^1 \cdot \mu_{M}^0$,
where $\mu_{M}^w, w \in \{0,1\}$ is the probability that the codewords of every edge $(i,j)$ in the particular cycle satisfy $[\mathbf{x}^w_{i}, \mathbf{x}^w_j] \in \mathcal{E}^{T^w}_{w,\epsilon}$ in block $\timeset^w$ by Eq. \eqref{muij}.

2) In block $\timeset^1$, if a particular cycle $(1,2,i_1,\ldots, i_{M-2})$ is constructed, it means $(\mathbf{x}^1_1,\mathbf{x}^1_2)$, $(\mathbf{x}^1_2,\mathbf{x}^1_{i_1})$ ,$\ldots$, $(\mathbf{x}^1_{i_{M-2}},\mathbf{x}^1_1)$ are all in $\mathcal{E}^{T^1}_{1,\epsilon}$. WLOG, let\rq{}s consider a particular 1-odd cycle (1, \ldots, M), the cycle is constructed if for any edge $\{i, \overline{i+1}\}$(where $\overline{i} \triangleq i \mod M$, if $i>M$), the number of times in $t \in \timeset^1$ that $(x_{i,t},x_{\overline{i+1},t})=(0,0)$ should be closed to $(1-p)^2q_{10}T$, i.e, $|\frac{1}{T}N((0,0)|[\mathbf{x}^1_{i}, \mathbf{x}^1_{\overline{i+1}}]) - (1-p)^2q_{10}| \leq \epsilon/4$ by Eq. \eqref{newrule1} for any $i$. Sum all of them, it suffices to show $|\frac{1}{MT}\sum_{i, t\in \timeset^1} \mathbf{1}((x_{i,t}, x_{{\overline{i+1}},t})=(0,0)) - (1-p)^2q_{10}|\leq \epsilon/4$ to upperbound $\mu_{M}^1$. 

	We can see the states $(x_{1,t}, x_{2,t}), (x_{2,t}, x_{3,t}), \ldots, (x_{M,t}, x_{1,t})$ are correlated, actually, the latter depends on the former(it seems $(x_{M,t}, x_{1,t})$ also depends on $(x_{1,t}, x_{2,t})$, but since $x_{1,t}, x_{2,t}$ are known when given $\channeloutput$, it only depends on the former $(x_{M-1,t},x_{M,t})$). A Markov chain $\mathcal{L}_t = [(x_{1,t}, x_{2,t}), (x_{2,t}, x_{3,t}), \ldots, (x_{M,t}, x_{1,t})]^\top$ with length $M$ can be constructed for any $t\in \timeset^1$, with the initial state $(x_{1,t}, x_{2,t})$ known, and the last states $(x_{M,t}, x_{1,t})$ partially known. Further, for different $t_1 \neq t_2$,  $\mathcal{L}_{t_1}$ and $\mathcal{L}_{t_2}$ are independent. Thus, we can construct a Markov chain with length $MT^1$ by assembling all $\{\mathcal{L}_t\}_{t \in \timeset^1}$ together, as shown in Fig. \ref{pfig15}, and denote it as $[h_n]$. It is a non-stationary Markov chain and the transition probability is easy to write.
\begin{figure}
\centering
\includegraphics[width=3.2in]{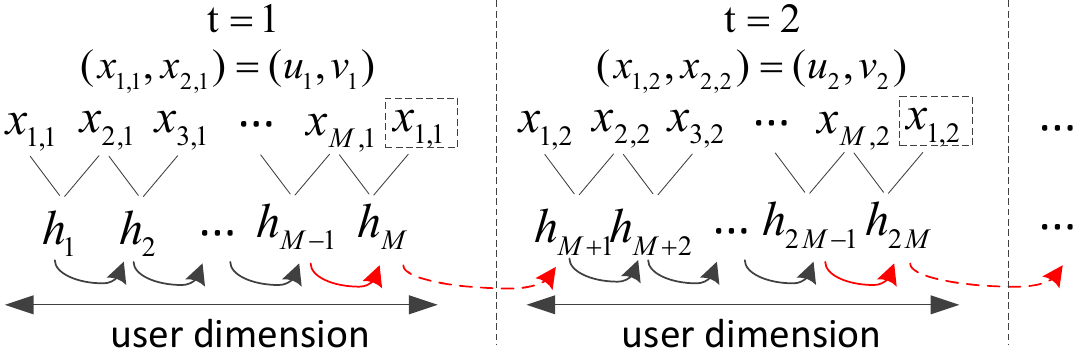}
\caption{Constructed Markov chain. Special transitions are drawn red.} \label{pfig15}
\end{figure}

	Then we need to estimate the deviation of the empirical average $\frac{1}{MT}N((0,0)|[h_n])$ from $(1-p)^2q_{10}$. By large deviation techniques for Markov chain\cite{Dembo2010}, we derive an upperbound of this probability:
\begin{align}
\mu_{M}^1 \leq \text{Pr}\left(\left|\frac{1}{MT}N((0,0)|[h_n])-(1-p)^2q_{10}\right|\leq \frac{\epsilon}{4}\right)
\leq e^{-(M-2)p_1 \varphi_1 T}
\end{align}
whose error exponent defined as $-\log(P)/(M-2)T$ is the first term in $C(p,q_{01},q_{10})$. Note that since $\mathbf{x}_1, \mathbf{x}_2$ are given, there are only $M-2$ degrees of freedom for the other $M-2$ codewords, which explains the $(M-2)$ factor here.

2) In block $\timeset^0$, the analysis is the same, now that the particular 1-odd cycle (1, \ldots, M) is constructed if $(\mathbf{x}^0_1,\mathbf{x}^0_2), (\mathbf{x}^0_2,\mathbf{x}^0_{3}) ,\ldots, (\mathbf{x}^0_{{M}},\mathbf{x}^0_1)$ are all in $\mathcal{E}^{T^0}_{0,\epsilon}$, which means $|\frac{1}{T}N((0,0)|[\mathbf{x}^0_{i}, \mathbf{x}^0_{\overline{i+1}}]) - (1-p)^2(1-q_{10})| \leq \epsilon/4$ for any $i$. Thus, we can construct the similar Markov chain $[h_n\rq{}]$ with length $MT^0$, and similarly, we will get the result:
\begin{align}
\mu_{M}^0 \leq \text{Pr}\left(\left|\frac{1}{MT}N((0,0)|[h_n\rq{}])-(1-p)^2(1-q_{10})\right|\leq \frac{\epsilon}{4}\right)\nonumber 
\leq e^{-(M-2)p_0 \varphi_0 T} \nonumber
\end{align}
whose error exponent results is the second term in $C(p,q_{01},q_{10})$, which completes the proof.

	
	The Markov Chain used in the proof shows the internal correlation structure of the partition problem, and it sheds light on a viable approach to solving more general cases with $K>2$ active users to be partitioned. 

\section{Necessary condition under the proposed sub-optimal method}\label{sec5}
	We have found a sufficient condition in terms of an upperbound on $\frac{T}{\log N}$ to achieve $P_e^{(\infty)} = 0$, another question to address is about a necessary condition to have $P_e^{(\infty)} = 0$ on $\frac{T}{\log N}$ under the same framework. We intend to show by randomly choosing elements of $\codebook$ with $p_{x}(1)=p$, no matter what $\epsilon$ we choose in the sub-optimal decoding method, if $\frac{T}{\log N} \leq 1/D(p,q_{10},q_{01}) - \xi$ for any $\xi>0$, we will have $P_e^{(\infty)} = 1$. If $p^* = \arg \max_p C(p,q_{10},q_{01})$, so that $\sufficientbound = 1/C(p^*,q_{10},q_{01})$, define $\necessarybound = 1/D(p^*,q_{10},q_{01})$, then we can see for $p_{x}(1)=p^*$, if $P_e^{(\infty)} = 0$, we should have $\frac{T}{\log N} \geq \necessarybound$. Comparison between $\sufficientbound$ and $\necessarybound$ shows the room where we could further improve the achievability result.

	Since for the proposed sub-optimal analysis, the decoding output is erroreous iff $\realactiveset \notin \hyperone$ or $\hyperone$ contains no 1-odd cycle. We will next derive $D(p,q_{10},q_{01})$ by considering when in $\hyperone$ there is definitely a 1-odd cycle of length 3. Denote $P_3 = \text{Pr}(\hyperone~\text{contains a 1-odd cycle of length 3, or}~\realactiveset \notin \hyperone)$, we have the following theorem.
	
\begin{thm}\label{thm2}
	\textit{When $K=2$, if elements of $\codebook$ are randomly generated by $p_x(1)=p$, for any constant $\xi>0$, if $\frac{T}{\log N} \leq 1/D(p,q_{10},q_{01}) -\xi$, no matter how $\epsilon$ is chosen, $P_3 \xrightarrow{N \to \infty} 1$ under the joint typical sequence based decoding, where 
\begin{align}
	D(p,q_{10},q_{01}) = p_1\varphi_1\rq{} + p_0 \varphi_0\rq{}, \label{rate}
\end{align}
and for $w \in \{0,1\}$,
\begin{align}
\varphi_w\rq{} = \frac{1}{p_w}\sup_{\lambda \in \mathbb{R}} \left(p_{y,x_1,x_2}(w,0,0)\log \frac{e^{2\lambda}}{p+(1-p)e^{2\lambda}} - 2p_{y,x_1,x_2}(w,1,0)\log \left((p+(1-p)e^{\lambda}\right)\right)\nonumber
\end{align}
}
\end{thm}
	We can see $D(p,q_{10},q_{01})$ also has the symmetry with $q_{10}, q_{01}$. The proof is put in Appendix \ref{app2}. The techniques we have employed in the proof are similar as those in the proof of Theorem \ref{thm1}.

\section{Comparison} \label{sec6}
	As stated in the introduction, our partition reservation has close relations to direct transmission and group testing. Since the average error considered in direction transmission system is not the same as the definition in this paper, we just compare with the group testing. It should be noted that group testing has a distinct objective, namely, learning the status of all users, rather an acceptable partitioning. 

	For $K=2$, by the tight achievable bound derived in Theorem IV. 1 in \cite{6157065}, we can see for group testing with $p_x(1)=p$, if $\frac{T}{\log N} \geq 1/C_g(p,q_{01},q_{10}) +\xi$ for any $\xi>0$, $P_e^{(\infty)} = 0$; on the other hand, if $P_e^{(\infty)} = 0$, we must have $\frac{T}{\log N} \geq 1/C_g(p,q_{01},q_{10})$, where
\begin{align}
	C_g(p,q_{01},q_{10}) = \min \left\{I(x_{1,t}; x_{2,t}, y_t), \frac{1}{2} I(x_{1,t}, x_{2,t}; y_t)\right\}\nonumber
\end{align}
$I(\cdot;\cdot)$ is the mutual information, $x_{1,t}, x_{2,t}$ are i.i.d. Bernoulli random variables with probability $p$, and $y_t$ is the random variable obeying condition distribution $p_{y|y_0}(y_t|x_{1,t}\vee x_{2,t})$ for given $x_{1,t}\vee x_{2,t}$. Define  $C_{g} \triangleq 1/\max_{p} C_g(p)$, it corresponds to the derived $C_1$ and $C_2$ in this paper.

	Let $q_{01}= q_{10}\triangleq q$, which is a symmetric binary channel, we obtain Fig. \ref{pfig7} (since $1/C$ might be zero, we plot $1/C$ instead of $C$ in the figure). It is shown for any $p, q$, $\sufficientbound < C_{g}$, which means less effort is needed for partition problem.  Also, we can see the derived lowerbound of the ratio $\necessarybound < \sufficientbound$, which shows there is still room to improve the achievable bound. We believe a sharper bound on $\sufficientbound$ by, e.g. counting cycles more carefully, should lead us to an improved result, which is one of our ongoing works at present. 
\begin{figure}
\centering
\includegraphics[width=3.5in]{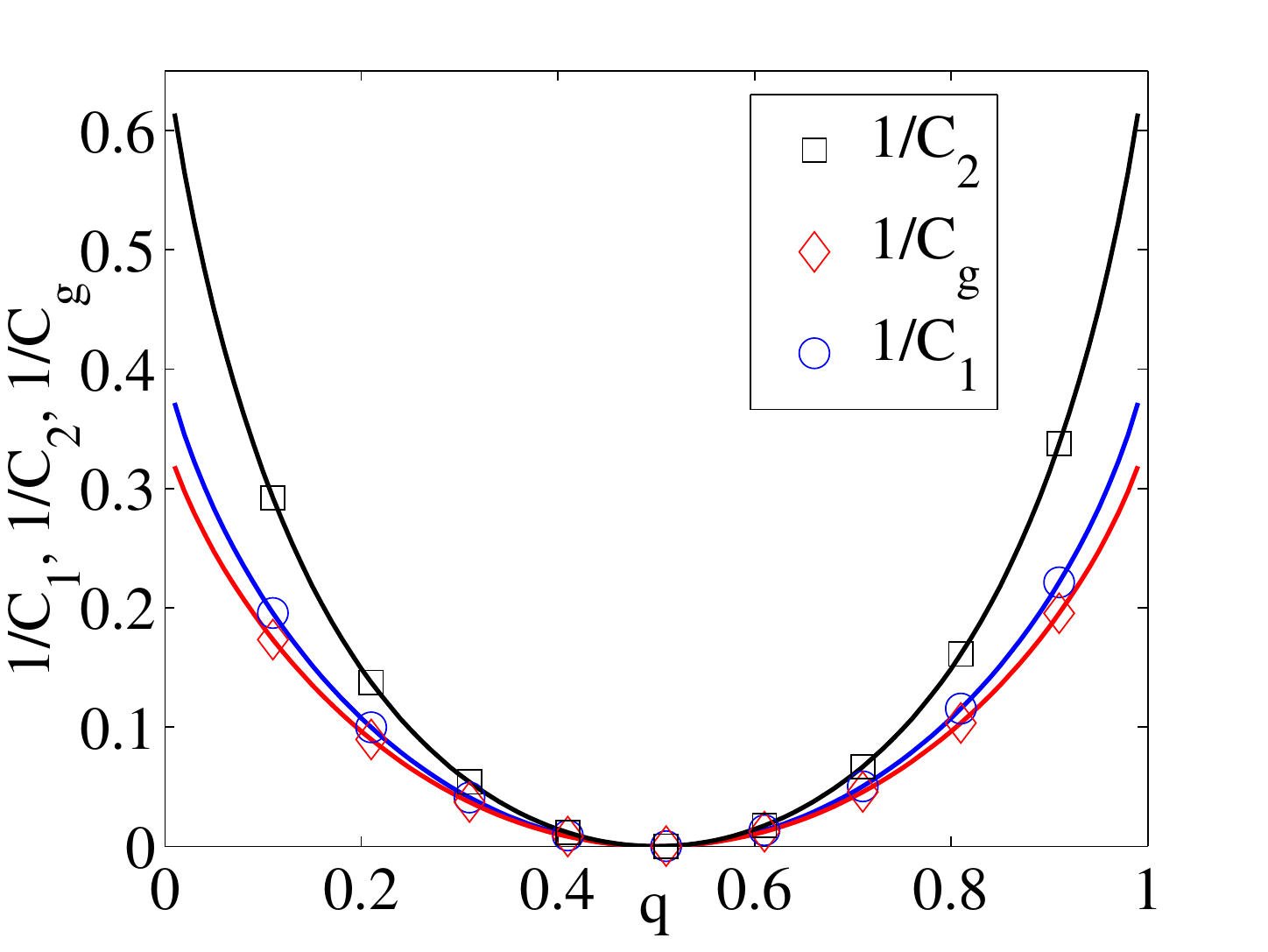}
\caption{Compare $\sufficientbound$, $C_{g}$ and $\necessarybound$.} \label{pfig7}
\end{figure}

\section{Conclusion}\label{sec7}
In this paper, we study a novel partitioning reservation problem over the noisy Boolean multi-access channels. We modify the $\hyperone$ construction process and the sequential decoding process in the noiseless case, and propose a general random coding approach and a sub-optimal jointly edge construction decoding method to obtain an achievable bound. A large deviation technique for non-stationary Markov chain is employed in the proof, which could be extended to study $K>2$ cases. Also, the derived achievable bound of $T \geq (\sufficientbound + \xi_1) \log N$ is better than the tight achievable bound in group testing. To study the tightness of this achievable bound, we also derive a necessary condition of $T$ so that if $T \leq ({\necessarybound}-\xi_2) \log N$, $P_e^{(\infty)} = 1$. It shows that the derived achievable bound is still able to be improved. The reason is that by union bound, there are too many overlapped odd cycles and thus we count vertices and edges repeatedly. This will be studied our in future works.

\appendices
\section{Proof of Theorem \ref{thm1}}\label{appendix.a}
\begin{IEEEproof}
	WLOG, we consider $0<p, q_{10}, q_{01}<1$ and $q_{10}+q_{01} \neq 1$, thus the continuity is guaranteed. We will show for any chosen $\epsilon>0$ of the sufficient constraints in \eqref{muij}(we use $\epsilon$ instead of $\tilde{\epsilon}$ for simplification of the notations), if $0 < \frac{\log N}{T}\leq C(p,q_{10},q_{01}) - \xi(\epsilon)$, there is always $P_e^{(\infty)} = 0$, where $\xi(\epsilon)>0$ is a function of $\epsilon$ and $\xi(\epsilon) \xrightarrow{\epsilon \to 0} 0$. Obviously, it is equivalent to the Theorem.

	Denotes $A\rq{}$ to be the event that $\hyperone$ contains 1-odd cycles, and $A \triangleq   A\rq{} \vee \left(\{1,2\} \notin E^{\rq{}}_T\right)$ (recall that $E^{\rq{}}_T$ is the set of edges of $\hyperone$), since by Eq. \eqref{muij}, Eq. \eqref{atypical},  $[\channeloutput, \mathbf{x}_1, \mathbf{x}_2] \in \mathcal{A}_{\epsilon}^T \Longrightarrow \{1,2\} \in E^{\rq{}}_T$, and $\channeloutput \in \mathcal{E}^{T}_{y,\epsilon}$, we have:
\begin{align}
	P_e^{(N)} \leq \text{Pr}(A) \leq \max_{[\channeloutput, \mathbf{x}_1, \mathbf{x}_2] \in \mathcal{A}_{\epsilon}^T}\text{Pr}(A\rq{}|[\channeloutput, \mathbf{x}_1, \mathbf{x}_2]) + \text{Pr}([\channeloutput, \mathbf{x}_1, \mathbf{x}_2] \notin \mathcal{A}_{\epsilon}^T)
\end{align}
	 By the features of strong typical set \cite{csiszar2011information}, we have $\text{Pr}\left([\channeloutput, \mathbf{x}_1, \mathbf{x}_2]\notin \mathcal{A}_{\epsilon}^{(T)}\right) \xrightarrow{T \to \infty} 0$. Thus, we just need to consider the probability of existence of any possible 1-odd cycle in $\hyperone$ on the condition that $[\channeloutput, \mathbf{x}_1, \mathbf{x}_2] \in \mathcal{A}_{\epsilon}^T$. In the following parts, this condition is assumed to be held.  For better statement, we denote $\timeset_{u,v}^w = \{t: 1\leq t \leq T, (y_t, x_{1,t}, x_{2,t}) = (w,u,v)\}$, and $T_{u,v}^w = |\timeset_{u,v}^w|$. Also, denote $\overline{T}_{u,v}^w = p_{y,x_1,x_2}(w,u,v)T$, and  $\overline{T}^w = p_w T$.

\subsection{Odd cycles for given $[\channeloutput, \mathbf{x}_1, \mathbf{x}_2] \in \mathcal{A}_{\epsilon}^T$}	
	For simplification, assume $N$ is an odd number. Consider any particular 1-odd cycle of length $M$, there are at most ${N-2 \choose M-2}(M-2)! \leq N^{M-2}$ such odd cycles out of $N$ nodes, and because of the symmetry in generating codewords, the existence of any of them in $\hyperone$ is equiprobable, thus WLOG we will select a particular one $\hyper_{e;M} \triangleq (1,2,\ldots, M)$. Denote $\tilde{p}_{M|[\channeloutput, \mathbf{x}_1, \mathbf{x}_2]}$ to be the probability that $\hyper_{e;M} \subseteq \hyperone$, and $P_{M|[\channeloutput, \mathbf{x}_1, \mathbf{x}_2]}$ the probability that $\hyperone$ contains 1-odd cycles of length $M$, by union bound, we have:
\begin{align}
	\max_{[\channeloutput, \mathbf{x}_1, \mathbf{x}_2] \in \mathcal{A}_{\epsilon}^T}\text{Pr}(A\rq{}|[\channeloutput, \mathbf{x}_1, \mathbf{x}_2])
\leq & \sum_{M=3,5, \ldots, N} \max_{[\channeloutput, \mathbf{x}_1, \mathbf{x}_2] \in \mathcal{A}_{\epsilon}^T} P_{M|[\channeloutput, \mathbf{x}_1, \mathbf{x}_2]} \label{PrE1}
\end{align}
and
\begin{align}
	P_{M|[\channeloutput, \mathbf{x}_1, \mathbf{x}_2]} \leq N^{M-2} \tilde{p}_{M|[\channeloutput, \mathbf{x}_1, \mathbf{x}_2]}
\end{align}
By Eq. \eqref{muij}, and note that we have already had $\channeloutput \in \mathcal{E}^{T}_{y,\epsilon}$, thus
\begin{align}
	\tilde{p}_{M|[\channeloutput, \mathbf{x}_1, \mathbf{x}_2]} = &\text{Pr}\left((1,2), (2,3) , \ldots, (M,1) \in E^{\rq{}}_T|[\channeloutput, \mathbf{x}_1, \mathbf{x}_2]\right)\nonumber\\
\leq &\text{Pr}\left([\mathbf{x}^1_{w}, \mathbf{x}^1_{\overline{w+1}}] \in \mathcal{E}^{T^1}_{1,\epsilon}, \forall w|[\channeloutput, \mathbf{x}_1, \mathbf{x}_2]\right) \text{Pr}\left([\mathbf{x}^0_{w}, \mathbf{x}^0_{\overline{w+1}}] \in \mathcal{E}^{T^0}_{0,\epsilon}, \forall w|[\channeloutput, \mathbf{x}_1, \mathbf{x}_2]\right)\nonumber\\
\triangleq &\mu^1_M \cdot  \mu^0_M,\label{maxmu}
\end{align}
which drives us to separately check the codewords in block $\timeset^0$ and $\timeset^1$.

\subsection{In $\timeset^1$}
	For $\mu^1_M \triangleq \text{Pr}\left([\mathbf{x}^1_{w}, \mathbf{x}^1_{\overline{w+1}}] \in \mathcal{E}^{T^1}_{1,\epsilon}, \forall w|[\channeloutput, \mathbf{x}_1, \mathbf{x}_2]\right)$, we can see the items $[\mathbf{x}^1_{w}, \mathbf{x}^1_{{\overline{w+1}}}]$ and $[\mathbf{x}^1_{w+1}, \mathbf{x}^1_{{\overline{w+2}}}]$ are correlated, so that we can\rq{}t separately calculate  $\text{Pr}\left([\mathbf{x}^1_{w}, \mathbf{x}^1_{{\overline{w+1}}}] \in \mathcal{E}^{T^1}_{1,\epsilon}|[\channeloutput, \mathbf{x}_1, \mathbf{x}_2]\right)$ for each $w$. However, we can see $[\mathbf{x}^1_{w+1}, \mathbf{x}^1_{{\overline{w+2}}}]$ only depends on the former item $[\mathbf{x}^1_{{w}}, \mathbf{x}^1_{{\overline{w+1}}}]$ , which inspires us to adopt Markov chain to model the problem. Before that, let\rq{}s further simply the calculation by summing up all the requirements of $[\mathbf{x}^1_{w}, \mathbf{x}^1_{\overline{w+1}}]$, $\forall 1\leq w \leq M$ to get a common one:
\begin{align}
	[\mathbf{x}^1_{w}, \mathbf{x}^1_{{\overline{w+1}}}] \in \mathcal{E}^{T^1}_{1,\epsilon}, \forall w \iff &\left|\frac{1}{T}N((0,0)|[\mathbf{x}^1_{w}, \mathbf{x}^1_{{\overline{w+1}}}] )-(1-p)^2q_{10}\right|\leq\epsilon/4, \forall w \nonumber\\
\Longrightarrow &\left|\frac{1}{MT}\sum_{w=1}^M \sum_{ t\in \timeset^1} \mathbf{1}((x_{w,t}, x_{{\overline{w+1}},t})=(0,0)) - (1-p)^2q_{10}\right|\leq \epsilon/4 \label{constraint1}
\end{align}
Thus, we suffice to find the probability that the pairs sequence $[(x^1_{{w},t}, x^1_{{\overline{w+1}},t})]_{1\leq w \leq M,t \in \timeset_1}$ satisfying Eq. \eqref{constraint1}, which is an upperbound of $\mu^1_M$.

\subsubsection{Markov chain}
	First, note that for any $t$, $(x_{1,t}, x_{2,t}), (x_{2,t}, x_{3,t}), \ldots, (x_{M,t}, x_{1,t})$ are correlated, and if $x_{1,t}, x_{2,t}$ are known, the latter depends only on the former: the first component should equal to the second one of the former, and the second component is generated randomly by Bernoulli pdf with probability $p$ (if the state is not $(x_{1,t}, x_{2,t})$ or $(x_{M,t}, x_{1,t})$). Thus, for any $t \in \timeset^1$, we can construct a Markov chain $\mathcal{L}_{t;x_{1,t}, x_{2,t}} = [(x_{1,t}, x_{2,t}), (x_{2,t}, x_{3,t}), \ldots, (x_{M,t}, x_{1,t})]$ of length $M$, whose initial state $(x_{1,t}, x_{2,t})$ is given, and the last state $(x_{M,t}, x_{1,t})$ is partially given. Second, for any two chain $\mathcal{L}_{t_1;x_{1,t_1}, x_{2,t_1}}$, $\mathcal{L}_{t_2;x_{1,t_2}, x_{2,t_2}}$, $t_1 \neq t_2$, they are independent. Which means we can concatenate them together to get a Markov chain with length $MT^1$. Because the initial state of each chain $\mathcal{L}_{t;x_{1,t}, x_{2,t}}$ is given by $(x_{1,t}, x_{2,t})$, and there are four possible values of $(x_{1,t}, x_{2,t}) \in \{0,1\}^2$, we will concatenate $\mathcal{L}_{t;x_{1,t}, x_{2,t}}$ according to the values of their initial states, by the order that first $(x_{1,t}, x_{2,t}) = (1,1)$, then $(1,0)$, $(0,1)$ and $(0,0)$ at last, as shown in Fig. \ref{pfig16}. 

\begin{figure}
\centering
\includegraphics[width=3.2in]{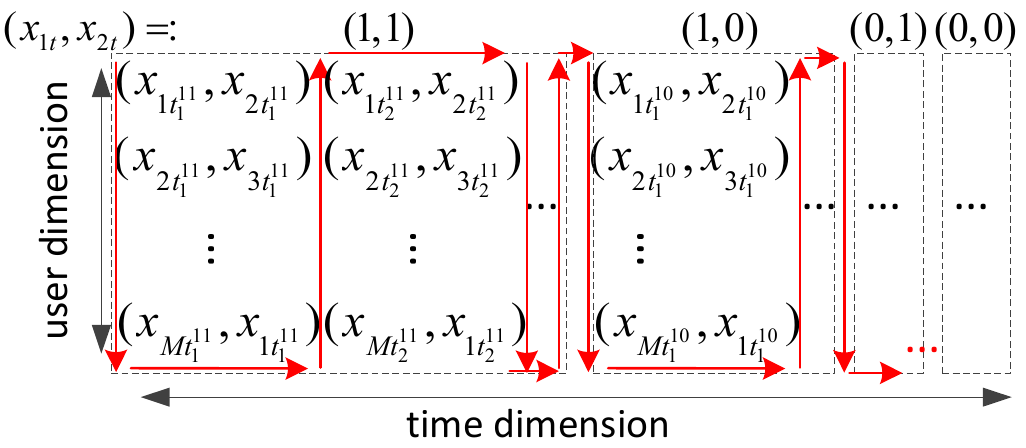}
\caption{Arrange $\{(x_{w,t}, x_{{\overline{w+1}},t})\}_{w, t}$ to a Markov chain. Here $t^{u,v}_j \in \timeset_{u,v}^1$; $t^{u_1,v_1}_i < t^{u_2,v_2}_j$ if $i<j$. The red arrow shows how the Markov chain is formed.} \label{pfig16}
\end{figure}

	Denote the obtained chain to be $[h_n]_{n=1}^{MT^1}$,  the Markov chain can be seen as a combination of a series of subsequences of length $M$, and can be partitioned into four blocks based on the values of the initial states of subsequences $(x_{1,t}, x_{\overline{2}, t})=(u,v) \in \{0,1\}^2$, each block has $T_{u,v}^1$ number of subsequences. Index possible values of $h_n$ to be $[(1,1), (1,0), (0,1), (0,0)]^\top = [h^0_i]_{i=1}^4$, we can define the transition matrix $\mathbf{\Pi}_{n+1|n}\triangleq [\pi_{n+1|n}(h^0_i|h^0_j)]_{i,j}$, where $\pi_{n+1|n}(h^0_i|h^0_j) \triangleq \text{Pr}(h_{n+1}=h^0_i|h_n=h^0_j)$. Thus, it is easy to write:
\begin{align}
	\mathbf{\Pi}_{n+1|n} = 
\begin{cases}
	\mathbf{T}_{u,v}, &\text{if}~n = k M,~ h_{n +1} =(u, v) \\
	\mathbf{C}_{u}, &\text{if}~n=k M-1,~h_{(k-1)M+1} =(u, v)\\
	\mathbf{\Pi}, &\text{otherwise}
\end{cases}
\label{transmatrix}
\end{align}
where $k$ is a natural number. The initial state $\pi_{1}((u,v)) = \mathbf{1}((u,v)=(1,1))$, and
\begin{align}
	\mathbf{T}_{11} \triangleq
	\begin{bmatrix}
	1 &1 &1 &1\\
	0 &0 &0 &0\\
	0 &0 &0 &0\\
	0 &0 &0 &0
	\end{bmatrix},~
	\mathbf{T}_{10} \triangleq
	\begin{bmatrix}
	0 &0 &0 &0\\
	1 &1 &1 &1\\
	0 &0 &0 &0\\
	0 &0 &0 &0
	\end{bmatrix},~
	\mathbf{T}_{01} \triangleq
	\begin{bmatrix}
	0 &0 &0 &0\\
	0 &0 &0 &0\\
	1 &1 &1 &1\\
	0 &0 &0 &0
	\end{bmatrix},
	\mathbf{T}_{00} \triangleq
	\begin{bmatrix}
	0 &0 &0 &0\\
	0 &0 &0 &0\\
	0 &0 &0 &0\\
	1 &1 &1 &1
	\end{bmatrix}\nonumber
\end{align}
\begin{align}
	\mathbf{C}_{1} \triangleq
	\begin{bmatrix}
	1 &0 &1 &0\\
	0 &0 &0 &0\\
	0 &1 &0 &1\\
	0 &0 &0 &0
	\end{bmatrix},~
	\mathbf{C}_{0} \triangleq
	\begin{bmatrix}
	0 &0 &0 &0\\
	1 &0 &1 &0\\
	0 &0 &0 &0\\
	0 &1 &0 &1
	\end{bmatrix},~
	\mathbf{\Pi} \triangleq
	\begin{bmatrix}
	p &0 &p &0\\
	1-p &0 &1-p &0\\
	0 &p &0 &p\\
	0 &1-p &0 &1-p
	\end{bmatrix}\nonumber
\end{align}
	This Markov chain is non-stationary, but except for the states transition related to $x_{1,t}, x_{2,t}$, we will get $\mathbf{\Pi}_{n+1|n} =\mathbf{\Pi}$ as the transition matrix, and actually $\mathbf{\Pi}$ will be the key factor in this problem. 

\subsubsection{Large deviation calculation}
	Now we just need to consider
\begin{align}
	\mu^1_M \leq \text{Pr}\left(\left|\frac{1}{MT^1}\sum_{n} \mathbf{1}(h_n = (0,0))- \frac{T}{T^1}(1-p)^2q_{10}\right|\leq \frac{T}{4T^1}\epsilon\right),
\end{align}
which is a problem of calculating a large deviation of the empirical means of a function($\mathbf{1}(h = (0,0))$) of a Markov chain $[h_n]$, and can be solved by the lemma in Chapter 2.3 in \cite{Dembo2010}. For completeness, we write the lemma here.

\begin{lem}
	\label{lemma1}
	\emph{For any random process $[{h}_{n}]_{n=1}^N$, for any function $f(\cdot)$  and the empirical means $W_{N} =\frac{1}{N}\sum_{n=1}^{N} f(h_n)$,
define the logarithmic moment generating function of $W_{N}$:
\begin{align}
	\Lambda_N(\lambda) \triangleq \log E[e^{\lambda W_N}].
\end{align}
If $\overline{x} = E[W_{N}]$ is finite, then for any non-empty closed interval $F=[a,b]$, we have:
\begin{align}
	\text{Pr}(W_{N} \in F) \leq 
\begin{cases}
	e^{-\Lambda_{N}^*(Na)}, &\text{if}~ \overline{x} < a\\
	e^{-\Lambda_{N}^*(Nb)}, &\text{if}~ \overline{x} > b\\
\end{cases}
\end{align}
where:
\begin{align}
	\Lambda^*_{N}(x) \triangleq \sup_{\lambda \in \mathbb{R}} \{\lambda x - \Lambda_N(\lambda)\}
\end{align}
is called the Fenchel-Legendre transform of $\Lambda_N(\lambda)$.}
\end{lem}

	Thus, the key problem of is to calculate the $\Lambda_{MT^1}(MT^1 \lambda)$, $\Lambda^*_{MT^1}(x)$ and $\overline{x} \triangleq E(W_{MT^1})$, where $W_{MT^1} =\frac{1}{MT^1}\sum_{n=1}^{MT^1} f(h_n)$ for our Markov chain $[h_n]$, with the function $f(h_n)=\mathbf{1}(h_n = (0,0))$. By direct calculation, we have the results:
\begin{lem}
\emph{
	For $[h_n]$ with transition matrix as \eqref{transmatrix}, we have
\begin{align}
	\Lambda_{MT^1}\left(MT^1\lambda\right) = &\sum_{(u,v)\in \{0,1\}^2} T_{u,v}^1\log (g_{u,v}(M,\lambda))
\end{align}
where
\begin{align}
	g_{1,1}(M,\lambda) = &J_{M-1}(\lambda)+(1-\alpha)J_{M-2}(\lambda)\nonumber\\
	g_{1,0}(M,\lambda) = g_{0,1}(M,\lambda) = &J_{M-1}(\lambda)\nonumber\\
	g_{0,0}(M,\lambda) = &e^{2\lambda}\left(J_{M-1}(\lambda)-p\left(1-e^{-\lambda}\right) J_{M-2}(\lambda)\right)
\end{align}
$\rho_{+}$, $\rho_{-}$ are the larger and smaller eigenvalues of $\mathbf{\Pi}$ (of rank 2), and
\begin{align}
	J_M(\lambda) = &\frac{\rho_{+}^M - \rho_{-}^M}{\rho_{+} - \rho_{-}}\\
	\alpha = &\rho_++\rho_-=p+(1-p)e^{\lambda}\label{alpha}\\
	\beta = &\rho_+-\rho_-= \sqrt{(2p-\alpha)^2+4p(1-p)}\label{beta}
\end{align}
When $\epsilon$ is sufficiently small,
\begin{align}
\overline{x}=&(1-p)^2 - \frac{3(1-p)^2}{M}+\frac{(1-p)(T_{1,1,0}+T_{1,0,1})+(3-2p)T_{1,0,0}}{MT^1} \nonumber\\
&\begin{cases}
	> \frac{T}{T^1}((1-p)^2q_{10}+\epsilon), &q_{01}+q_{10}<1;\\
	<\frac{T}{T^1}((1-p)^2q_{10}-\epsilon), &q_{01}+q_{10}>1.
\end{cases}
\end{align}
}
\end{lem}

	So with the lemmas above, we can bound $\mu_M^1$ by:
\begin{align}
	\mu_M^1 \leq &\text{Pr}\left(W_{MT^1} \in \frac{T}{T^1}\left[(1-p)^2q_{10}-\epsilon, (1-p)^2q_{10}+\epsilon\right]\right)\nonumber\\
\leq&\begin{cases}
	e^{-\Lambda_{MT^1}^*\left(MT((1-p)^2q_{10}+\epsilon)\right)}, &q_{01}+q_{10}<1;\\
	e^{-\Lambda_{MT^1}^*\left(MT ((1-p)^2q_{10}-\epsilon)\right)}, &q_{01}+q_{10}>1.
\end{cases}
\label{originalmu}
\end{align}
To get the expression in the theorem, we will use the continuity to drop the $\epsilon$, and substitute all $T^1$, $T_{u,v}^1$ by $\overline{T}^1$ and $\overline{T}_{u,v}^1$, which is the value corresponding to $\epsilon = 0$, and use the lemma belows to make the result more concise as well.

\subsubsection{Concise solution}
	To get a concise solution, we give a further result:
\begin{lem}
	\emph{
	$\forall M\geq 3$, if $\epsilon=0$, which means $T_{u,v}^1 = \overline{T}_{u,v}^1$, $T^1=\overline{T}^1$,  then for $\overline{x}^1 \triangleq \frac{(1-p)^2 q_{10}}{p_1}$, we have 
\begin{align}
	\frac{1}{M\overline{T}^1}\Lambda_{M\overline{T}^1}^*\left(MT((1-p)^2q_{10})\right)
=&\sup_{\lambda \in \mathbf{R}}\left\{ \lambda \overline{x}^1 - \frac{1}{M\overline{T}^1} \Lambda_{M\overline{T}^1}\left(M\overline{T}^1\lambda\right)\right\} \nonumber\\
 \geq &\frac{M-2}{M}\sup_{\lambda\in \mathbf{R}}(\lambda \overline{x}^1 - \log \rho_+)
\end{align}}
\end{lem}
\begin{proof}
	Define
\begin{align}
	F(\lambda) \triangleq  &\lambda \overline{x}^1 - \frac{1}{M\overline{T}^1} \Lambda_{M\overline{T}^1}\left(M\overline{T}^1\lambda\right)\nonumber\\
	G(\lambda) \triangleq&\frac{M-2}{M}(\lambda \overline{x}^1-\log \rho_{+})\nonumber
\end{align}
Due to Lemma 2.2.5 in \cite{Dembo2010}, $F(\lambda)$ and $G(\lambda)$ are both convex. And we can find by calculation if $q_{10}+q_{01} <1$, $F\rq{}(0), G\rq{}(0) < 0$;  if $q_{10}+q_{01} > 1$, $F\rq{}(0), G\rq{}(0) > 0$. Which means the maximum of $F(\lambda), G(\lambda)$ will locate at $\lambda <0$, for $q_{10}+q_{01} <1$, and vice versa. Now we consider $q_{10}+q_{01} <1$, the other case is the same. We will show if $\lambda^*=\arg\max_{\lambda} G(\lambda)$, for $\lambda^*\leq \lambda \leq 0$, there is always $F(\lambda) \geq G(\lambda)$, which implies the lemma.

 The convexity of $G(\lambda)$ means $G\rq{}(\lambda^*) =0$, by calculating $G\rq{}(\lambda)$, we have:
\begin{align}
	\lambda^* \leq \lambda \leq 0 \iff &G\rq{}(\lambda) \leq 0, ~\lambda \leq 0\nonumber\\
	\iff&e^{\lambda} \geq \frac{\overline{x}^1\rho_{+}(2\rho_+ - p)}{(1-p)((1+\overline{x}^1)\rho_+ - p)}, \lambda \leq 0 \label{calcon1}
\end{align}
	Also since the convexity of $\log (\cdot)$, we have:
\begin{align}
-2 \lambda \overline{x}^1 + \sum_{u,v}\frac{\overline{T}_{u,v}^1\log (g_{u,v})}{\overline{T}^1} 
\leq &\log \left((1-\overline{x}^1)\frac{\sum_{(u,v)\neq (0,0)}\overline{T}_{u,v}^1 g_{u,v}}{\sum_{(u,v)\neq (0,0)}\overline{T}_{u,v}^1}+\overline{x}^1\frac{g_{0,0}}{e^{2\lambda}}\right)\label{calcon2}
\end{align}
By \eqref{calcon1} and \eqref{calcon2}, we can derive by calculation that $F(\lambda) \geq G(\lambda), \forall M \geq 3$, which completes the proof for $q_{01}+q_{10}<1$. It is the similar for $q_{01}+q_{10}>1$.
\end{proof}

	\subsubsection{Continuity to drop $\epsilon$}
	By continuity,  assume
$$\sup_{[\channeloutput, \mathbf{x}_1, \mathbf{x}_2] \in \mathcal{A}_{\epsilon}^T}\left|\Lambda_{MT^1}^*\left(MT ((1-p)^2q_{10}-\epsilon)\right)-\Lambda_{M\overline{T}^1}^*\left(MT((1-p)^2q_{10})\right)\right|=\frac{\xi(\epsilon)}{4} >0 $$
we can see $\xi(\epsilon) \xrightarrow{\epsilon \to 0} 0$. Note that $\xi(\epsilon)$ is also dependent with $p, q_{10}, q_{01}$, but we don\rq{}t write it explicitly for simplification. Thus, we can derive from \eqref{originalmu} and Lemma 3 that when $q_{01}+q_{10}\neq 1$, 
\begin{align}
	\mu_M^1 \leq e^{ -(M-2)p_1T\left(\sup_{\lambda \in \mathbf{R}}\left(\lambda \overline{x}^1-\log \rho_{+}\right) -  \xi(\epsilon)/4\right)} \label{T1}
\end{align}

\subsection{In $\timeset^0$}

		For $\mu^0_M \triangleq \text{Pr}\left([\mathbf{x}^0_{w}, \mathbf{x}^0_{\overline{w+1}}] \in \mathcal{E}^{T^0}_{0,\epsilon}, \forall w|[\channeloutput, \mathbf{x}_1, \mathbf{x}_2]\right)$, since the symmetry, we can solve the problem directly from $\timeset^1$ case. Consider another noisy environment where $q_{10}\rq{} = 1-q_{10}$, and $q_{01}\rq{} = 1-q_{01}$. Then easily we can see $\forall w \in \{0,1\}, p_w\rq{} = 1-p_w$. For any given $\codebook$ satisfying $[\channeloutput, \mathbf{x}_1, \mathbf{x}_2] \in \mathcal{A}^T_{\epsilon}$ and $\forall w, [\mathbf{x}^1_{w}, \mathbf{x}^1_{\overline{w+1}}] \in \mathcal{E}_{e}^{\left(T^1\right)\rq{}}$, it is equivalent to that $\forall w, [\mathbf{x}^0_{w}, \mathbf{x}^0_{\overline{w+1}}] \in \mathcal{E}^{T^1}_{1,\epsilon}$. Thus
by the definition of $\mu_M^w$ in Eq. \eqref{maxmu}, we have
\begin{align}
\max_{[\channeloutput, \mathbf{x}_1, \mathbf{x}_2] \in \mathcal{A}^T_{\epsilon}}\left(\mu^1_M\right)\rq{}
=\max_{[\channeloutput, \mathbf{x}_1, \mathbf{x}_2] \in \mathcal{A}^T_{\epsilon}}\mu^0_M\nonumber
\end{align}
By the upperbound derived for $\mu_M^1$ in Eq. \eqref{T1}, we have:
\begin{align}
	\mu_M^0 \leq e^{ -(M-2)p_0T\left(\sup_{\lambda \in \mathbf{R}}\left(\lambda \overline{x}^0-\log \rho_{+}\right) -  \xi(\epsilon)/4\right)} \label{T0}
\end{align} 
where $\overline{x}^0 \triangleq \frac{(1-p)^2 (1-q_{10})}{p_0}$.

	\subsection{Complete the proof}
	From Eq. \eqref{T0} and Eq. \eqref{T1}, we can upperbound $\tilde{p}_{M|[\channeloutput, \mathbf{x}_1, \mathbf{x}_2]}$ in Eq. \eqref{maxmu}:
\begin{align}
	\max \tilde{p}_{M|[\channeloutput, \mathbf{x}_1, \mathbf{x}_2]} \leq e^{ -(M-2)T\left(p_1 \varphi_1 + p_0 \varphi_3-  \xi(\epsilon)/2\right)}\nonumber
\end{align}
Thus, if $\frac{\log N}{T} \leq p_1 \varphi_1 + p_0 \varphi_0 - \xi(\epsilon)$, we have
\begin{align}
	P_e^{(N)} 
\leq &\sum_{M =3, 5, \ldots} e^{(M-2)T\left( \log N /T - (p_1\varphi_1+ p_0 \varphi_0) - \xi(\epsilon)/2\right)}+ \text{Pr}([\channeloutput, \mathbf{x}_1, \mathbf{x}_2] \notin \mathcal{A}_{\epsilon}^T)\nonumber\\
\leq  &\sum_{M =3, 5, \ldots} e^{-(M-2)\frac{\xi(\epsilon)}{2}T}+ \text{Pr}([\channeloutput, \mathbf{x}_1, \mathbf{x}_2] \notin \mathcal{A}_{\epsilon}^T) \xrightarrow{T \to \infty} 0\nonumber
\end{align}
which completes the proof.
\end{IEEEproof}
\section{Proof of Theorem 2}\label{app2}
\begin{IEEEproof}
	Denote $A_3$ the event that $\hyperone$ contains 1-odd cycle of length 3, for any $\epsilon$ in the necessity of Eq. \eqref{muij}(we use $\epsilon$ instead of $\hat{\epsilon}$ for simplifying the notations), we have by considering the strong typical set:
\begin{align}
	P_3 = &\text{Pr}((A_3\vee \realactiveset\notin \hyperone),  [\channeloutput,\mathbf{x}_1,\mathbf{x}_2]\in \mathcal{A}_{\epsilon}^T) + \text{Pr}([\channeloutput,\mathbf{x}_1,\mathbf{x}_2]\notin \mathcal{A}_{\epsilon}^T)\nonumber\\
\geq &\text{Pr}(A_3,  [\channeloutput,\mathbf{x}_1,\mathbf{x}_2]\in \mathcal{A}_{\epsilon}^T)\label{typical3}
\end{align} 
It can be seen that Eq. \eqref{typical3} is non-decreasing with $\epsilon$, thus equivalently, we just need to show if for any $\frac{\log N}{T} \geq \necessarybound + \xi(\epsilon)$, where $\xi(\epsilon)>0$ is a function satisfying $\xi(\epsilon) \xrightarrow{\epsilon \to 0} 0$, we have $\text{Pr}(A_3,  [\channeloutput,\mathbf{x}_1,\mathbf{x}_2]\in \mathcal{A}_{\epsilon}^T) \xrightarrow{T \to \infty} 1$. Further, since
\begin{align}
	\text{Pr}(A_3,  [\channeloutput,\mathbf{x}_1,\mathbf{x}_2]\in \mathcal{A}_{\epsilon}^T)\geq \text{Pr}([\channeloutput,\mathbf{x}_1,\mathbf{x}_2]\in \mathcal{A}_{\epsilon}^T) \min_{[\channeloutput,\mathbf{x}_1,\mathbf{x}_2]\in \mathcal{A}_{\epsilon}^T}\text{Pr}(A_3|[\channeloutput,\mathbf{x}_1,\mathbf{x}_2]) \label{typical4}
\end{align}
and for any $\epsilon > 0$, $\text{Pr}([\channeloutput,\mathbf{x}_1,\mathbf{x}_2]\in \mathcal{A}_{\epsilon}^T) \xrightarrow{T \to \infty} 1$, so $\min_{[\channeloutput,\mathbf{x}_1,\mathbf{x}_2]\in \mathcal{A}_{\epsilon}^T}\text{Pr}(A_3|[\channeloutput,\mathbf{x}_1,\mathbf{x}_2])$ is considered in the following part, and $[\channeloutput,\mathbf{x}_1,\mathbf{x}_2]\in \mathcal{A}_{\epsilon}^T$ is assumed to be held. Since
\begin{align}
\text{Pr}(A_3|[\channeloutput,\mathbf{x}_1,\mathbf{x}_2])\triangleq &\text{Pr}\left(\codebook:\bigcup_{i=3}^N \left(\hyperone ~\text{contains cycle}~ (1,2,i)\right)|[\channeloutput,\mathbf{x}_1,\mathbf{x}_2]\right)\nonumber\\
\overset{(a)}{=}&1-\left(1-\text{Pr}\left(\mathbf{x}_3: \hyperone ~\text{contains cycle}~ (1,2,3)|[\channeloutput,\mathbf{x}_1,\mathbf{x}_2]\right)\right)^{N-2}\nonumber
\end{align}
Eq. $(a)$ is because codewords of users are i.i.d. generated, So WLOG we just consider the cycle $(1,2,3)$ and the codeword $\mathbf{x}_3$. By the sub-optimal jointly decoding method and since $\channeloutput \in \mathcal{E}^{T}_{y,\epsilon}$ is already held, it means $\forall w \in\{0,1\}, [\mathbf{x}^w_1,\mathbf{x}^w_2], [\mathbf{x}^w_2,\mathbf{x}^w_3], [\mathbf{x}^w_1,\mathbf{x}^w_3] \in \mathcal{E}_{\epsilon}^{T^w}$. Because $[\channeloutput,\mathbf{x}_1,\mathbf{x}_2]\in \mathcal{A}_{\epsilon}^T$ means $[\mathbf{x}^w_1,\mathbf{x}^w_2] \in \mathcal{E}_{\epsilon}^{T^w}$, only $[\mathbf{x}^w_2,\mathbf{x}^w_3]$ and $ [\mathbf{x}^w_1,\mathbf{x}^w_3]$ should be considered.
 
	Time is separated to 8 parts $\timeset_{u,v}^w$ by the values of $(y_t, x_{1,t}, x_{2,t})=(w,u,v)$, denote $a^w_{u,v} \triangleq |\{t\in \timeset_{u,v}^w:x_{3,t}=0\}|$ to be the number of slots that $x_{3, t}$ takes value zeros, we have $0\leq a_{u,v}^w \leq T_{u,v}^w$, and  $\forall w \in \{0,1\}$,
\begin{align}
\begin{cases}
	N((0,0)|[\mathbf{x}^w_1,\mathbf{x}^w_3]) = a^w_{0,0}+a^w_{0,1}\\
	N((0,0)|[\mathbf{x}_2^w,\mathbf{x}^w_3]) = a^w_{0,0}+a^w_{1,0}
\end{cases}
\end{align}
Then by the definition of $\mathcal{E}_{\epsilon}^{T^w}$ in Eq. \eqref{newrule1}, we have:
\begin{align}
	\begin{cases}
	\left|\frac{1}{T}(a^w_{0,0}+a^w_{0,1})-p_{y,y_0}(w,0)\right|\leq \frac{\epsilon}{4}\\
	\left|\frac{1}{T}( a^w_{0,0}+a^w_{1,0})-p_{y,y_0}(w,0)\right|\leq \frac{\epsilon}{4}
\end{cases}
\end{align}
Up to a constant factor of $\epsilon$, it is equivalent to:
\begin{align}
	\begin{cases}
	\left|\frac{1}{T}(a^w_{0,0}+a^w_{0,1})-p_{y,y_0}(w,0)\right|\leq \frac{\epsilon}{4}\\
	\frac{\left|a^w_{1,0}-a^w_{0,1}\right|}{T}\leq \frac{\epsilon}{4}
\end{cases}
\end{align}
Denote $\mathbf{b}^w \triangleq \frac{1}{T^w}[a^w_{0,0}, a^w_{1,0},a^w_{0,1}]^\top$, and:
\begin{align}
	\mathcal{B}_{\epsilon}^{T^w} = 
\left\{\begin{array}{ll}
[b^w_{0,0},b^w_{1,0},b^w_{0,1}]^\top\in \mathbb{R}^3:&\\
\left|\frac{T^w}{T}(b^w_{0,0}+b^w_{0,1})-p_{y,y_0}(w,0)\right|\leq \frac{\epsilon}{4};&\\
	\frac{T^w\left|b^w_{1,0}-b^w_{0,1}\right|}{T}\leq \frac{\epsilon}{4};&\\
	0\leq b^w_{u,v} \leq \frac{T_{u,v}^w}{T^w}.
\end{array}\right\}
\end{align}
Thus,
\begin{align}
	\text{Pr}(A_3|[\channeloutput,\mathbf{x}_1,\mathbf{x}_2]) = &1- \left( 1-\prod_{w=1}^2\text{Pr}\left(\mathbf{b}^w\in \mathcal{B}_{\epsilon}^{T^w}|[\channeloutput,\mathbf{x}_1,\mathbf{x}_2]\right)\right)^{N-2}\nonumber\\
\triangleq &1-(1-\mu_3^0 \cdot \mu_3^1)^{N-2} \label{52}
\end{align}
We can still use large deviation technique to calculate $\mu_3^w$. Now the distribution of $a^w_{u,v}$ is a binomial distribution:
\begin{align}
	\forall 0\leq n \leq T_{u,v}^w, ~\text{Pr}(a^w_{u,v}=n) = {T_{u,v}^w \choose n} (1-p)^n p^{T_{u,v}^w-n}
\end{align}
and $a^w_{u,v}$ are independent with each other. So the logarithmic moment generating function of $\mathbf{b}^w$:
\begin{align}
	\Lambda_w(\boldsymbol{\lambda})\triangleq &
\lim_{T^w \to \infty}\frac{1}{T^w}\log E\left[e^{\boldsymbol{\lambda}^\top \mathbf{b}^w T^w}\right]\nonumber\\
=&\lim_{T^w \to \infty}\frac{1}{T^w} \log \prod_{u\wedge v = 0}\sum_{a^w_{u,v}=0}^{T_{u,v}^w}{T_{u,v}^w \choose a^w_{u,v}} (1-p)^{a^w_{u,v}} p^{T_{u,v}^w-a^w_{u,v}}e^{\lambda_{u,v} a^w_{u,v}}\nonumber\\
=&\sum_{u\wedge v = 0} \beta^w_{u,v} \log \left(p+(1-p)e^{\lambda_{u,v}}\right)
\end{align}
where $\beta^w_{u,v} = \lim_{T^w \to \infty}\frac{T^w_{u,v}}{T^w}$, and $\boldsymbol{\lambda} = [\lambda_{1,0}, \lambda_{0,1},\lambda_{0,0}]^\top$. It is easy to see we have the origin belongs to the interior of $\mathcal{D}_{w,\boldsymbol{\lambda}}\triangleq \{\boldsymbol{\lambda} \in \mathbb{R}^3: \Lambda_w(\boldsymbol{\lambda})<\infty\}$, by the G{\"a}rtner-Ellis Theorem 2.3.6 in \cite{Dembo2010}, we have:
\begin{align}
\liminf_{T^w \to \infty}\frac{1}{T^w} \log \text{Pr}\left(\mathbf{b}^w \in \mathcal{B}_{\epsilon}^{T^w}|[\channeloutput,\mathbf{x}_1,\mathbf{x}_2]\right) \geq -\inf_{\mathbf{b}^w\in \left(\mathcal{B}_{\epsilon}^{T^w}\right)^o \cap \mathcal{F}^w} \Lambda_{w,\epsilon}^*(\mathbf{b}^w) \triangleq -\underline{\varphi}_{w,\epsilon}\label{lowerboundl}
\end{align}
where 
\begin{align}
	\Lambda_{w,\epsilon}^*(\mathbf{b}^w) =  \sup_{\boldsymbol{\lambda} \in \mathbb{R}^3} \left[\boldsymbol{\lambda}^\top \mathbf{b}^w-\sum_{u\wedge v = 0} \beta^w_{u,v} \log \left(p+(1-p)e^{\lambda_{u,v}}\right)\right]
\end{align}
$\left(\mathcal{B}_{\epsilon}^{T^w}\right)^o$ is the interior of $\mathcal{B}_{\epsilon}^{T^w}$, and $\mathcal{F}^w$ is the exposed point of $\Lambda_{w,\epsilon}^*(\mathbf{b}^w)$. Further, since for any $[b_{1,0},b_{0,1},b_{0,0}]^{\top} \in \left(\mathcal{B}_{\epsilon}^{T^w}\right)^o$, we can always find some $[\lambda_{1,0},\lambda_{0,1},\lambda_{0,0}]^{\top} \in \mathcal{D}_{w,\lambda}^o$, so that
\begin{align}
	[b_{1,0},b_{0,1},b_{0,0}]^{\top} = &\nabla \Lambda_w([\lambda_{1,0},\lambda_{0,1},\lambda_{0,0}]^{\top})\nonumber\\
=&\left[\beta_{u,v}^w\frac{(1-p)e^{\lambda_{u,v}}}{p+(1-p)e^{\lambda_{u,v}}}\right]^{\top}_{u\wedge v = 0}
\end{align}
by Lemma 2.3.9 in \cite{Dembo2010}, we have $[b_{1,0},b_{0,1},b_{0,0}]^{\top} \in \mathcal{F}^w$, so $\left(\mathcal{B}_{\epsilon}^{T^w}\right)^o \subseteq\mathcal{F}^w$, which means we can calculate the infimum in the lower bound \eqref{lowerboundl} in $\left(\mathcal{B}_{\epsilon}^{T^w}\right)^o$. 

	Because of the continuity, assume the radius of the neighbourhood of $\underline{\varphi}_{w,\epsilon}|_{\epsilon=0}$ to be:
\begin{align}
	\sup_{[\channeloutput,\mathbf{x}_1,\mathbf{x}_2]\in \mathcal{A}_{\epsilon}^T}\left|\underline{\varphi}_{w,\epsilon}-\inf_{\mathbf{b}^w\in \left(\mathcal{B}_w\right)^o} \Lambda^*_w(\mathbf{b}^w)\right| = \xi_1(\epsilon)\label{58}
\end{align}
where 
\begin{align}
	\Lambda_{w}^*(\mathbf{b}^w) =  \sup_{\boldsymbol{\lambda} \in \mathbb{R}^3} \left[\boldsymbol{\lambda}^\top \mathbf{b}^w-\sum_{u\wedge v = 0} \frac{p_{y,x_1,x_2}(w,u,v)}{p_w} \log \left(p+(1-p)e^{\lambda_{u,v}}\right)\right]
\end{align}
\begin{align}
	\mathcal{B}_w =
\left\{\begin{array}{ll}
\mathbf{b}^w =[b^w_{0,0},b^w_{1,0},b^w_{0,1}]^\top\in \mathbb{R}^3:&\\
0 \leq b^w_{0,0}\leq \frac{p_{y,y_0}(w,0)}{p_w}, ~0 \leq b^w_{1,0}, b^w_{0,1}\leq \frac{p_{y,x_1,x_2}(w,1,0)}{p_w}&\\
b^w_{1,0}=b^w_{0,1} = \frac{p_{y,y_0}(w,0)}{p_w} - b^w_{0,0}.
\end{array}\right\}
\end{align}
are $\Lambda_{w,\epsilon}^*(\mathbf{b}^w)\mid_{\epsilon=0}$, $\mathcal{B}_{\epsilon}^{T^w}\mid_{\epsilon=0}$ respectively.

Denote $\hat{D}(p,q_{10},q_{01}) = p_1 \inf_{\mathbf{x}^1\in \left(\mathcal{B}_1\right)^o} \Lambda^*_1(\mathbf{x}^1) + p_0 \inf_{\mathbf{x}^0\in \left(\mathcal{B}_0\right)^o} \Lambda^*_0(\mathbf{x}^0)$, by Eq. \eqref{lowerboundl}, \eqref{58}, asymptotically we have
\begin{align}
\mu_3^0 \mu_3^1
\geq &e^{-T\left(\left(p_1+\frac{\epsilon}{2}\right)\left(\inf_{\mathbf{x}^1\in \mathcal{B}_1} \Lambda_1^*(\mathbf{x}^1)+\xi_1(\epsilon)\right)+\left(p_0+\frac{\epsilon}{2}\right)\left(\inf_{\mathbf{x}^0\in \mathcal{B}_0} \Lambda_0^*(\mathbf{x}^0)+\xi_1(\epsilon)\right)\right)}\nonumber\\
=&\text{exp}\left(-T\left(\hat{D}(p,q_{10},q_{01}) + \xi(\epsilon)/2 \right)\right) \label{93}
\end{align}
where $\xi(\epsilon)>0$ and $\xi(\epsilon) \xrightarrow{\epsilon\to 0} 0$. Which means there is a $c$ so that $\mu_3^0 \mu_3^1 = e^{-cT}$ and $c \leq \hat{D}(p,q_{10},q_{01}) + \xi(\epsilon)/2$. 
If $\frac{\log (N-2)}{T} \geq \hat{D}(p,q_{10},q_{01}) + \xi(\epsilon)$, by Eq. \eqref{52}, we can see 
\begin{align}
	\text{Pr}(A_3|[\channeloutput,\mathbf{x}_1,\mathbf{x}_2]) = &
1-\left[\left(1-e^{-cT}\right)^{e^{cT}}\right]^{e^{T\left(\frac{\log (N-2)}{T}-c\right)}}\nonumber\\
\geq &1-\left[\left(1-e^{-cT}\right)^{e^{cT}}\right]^{e^{\frac{\xi(\epsilon)}{2}T}} \xrightarrow{T \to \infty} 1 \nonumber
\end{align}
which is what we intend to show, except that it is $\hat{D}(p,q_{10},q_{01})$ here. 

	So the last step is to show $\hat{D}(p,q_{10},q_{01}) = D(p,q_{10},q_{01})$ in the theorem, which is equivalent to showing $\varphi_w\rq{} = \inf_{\mathbf{x}^w\in \mathcal{B}_w} \Lambda^*_w(\mathbf{b}^w)$. We just calculate $\varphi_1\rq{}=\inf_{\mathbf{b}^1\in \mathcal{B}_1} \Lambda^*_1(\mathbf{b}^1)$, the other part $\varphi_0\rq{}$ can be easily proved by symmetry. It can be seen if the inf-sup is achieved when $\lambda_{1,0}=\lambda_{0,1}=\lambda_{0,0}/2$, the equation is held. Due to the convexity of $\Lambda^*_w(\mathbf{b}^w)$ with both $\lambda \in \mathbb{R}^3$ and $\mathbf{b}^w \in \mathcal{B}_w$ (by checking their second derivatives), we first let $\forall (u,v), \frac{\partial \Lambda^*_1(\mathbf{b}^1)}{\partial \lambda_{u,v}}=0$ to find the optimal $\boldsymbol{\lambda}^*$ for any given $\mathbf{b}^1=[b_{1,0}^1, b_{0,1}^1,b_{0,0}^1]$,
\begin{align}
	\lambda_{1,0}^* =\lambda_{0,1}^* = &\log \frac{p}{1-p}\frac{a_{1,0}^1}{p_{x_1,x_2|y}(1,0|1)-a_{1,0}^1}\\
	\lambda_{0,0}^* = &\log \frac{p}{1-p}\frac{a_{0,0}^1}{p_{x_1,x_2|y}(0,0|1)-a_{0,0}^1}
\end{align} 
Substitute $\boldsymbol{\lambda}^*$ to $\Lambda^*_1(\mathbf{b}^1)$, and let $\frac{\partial\Lambda^*_1(\mathbf{b}^1)}{\partial a_{0,0}^1} =0$, by using the requirement $a_{1,0}^1+a_{0,0}^1 = p_{x_1,x_2|y}(0,0|1)$, we can obtain the equation $2 \lambda_{1,0}^* = \lambda_{0,0}^*$.
Thus, when the inf-sup is achieved, we have $\lambda_1=\lambda_2=\lambda_3/2$, which completes the proof.
\end{IEEEproof}
\bibliographystyle{IEEEtran}
\bibliography{toword}

\end{document}